\documentclass[11pt,reqno,a4paper]{amsart}
\pdfoutput=1
\usepackage{graphicx}
\usepackage{amssymb}
\usepackage{amsthm}
\usepackage{amsmath}
\usepackage[left=0.9in,right=0.9in,top=1.2in,bottom=1.2in]{geometry}
\usepackage{mathtools}
\usepackage{caption}
\usepackage{array}
\usepackage{hyperref}
\usepackage{xcolor}
\hypersetup{
    colorlinks,
    linkcolor={red!50!black},
    citecolor={blue!50!black},
    urlcolor={blue!80!black}
}

\newcommand{\arxiv}[1]{\href{http://arxiv.org/abs/#1}{\texttt{arXiv:#1}}}

\newtheorem{lemma}{Lemma}
\newtheorem{theorem}{Theorem}
\newtheorem{proposition}{Proposition}
\newtheorem{corollary}{Corollary}
\newtheorem{remark}{Remark}

\newcommand{\R}{\mathbb{R}}
\newcommand{\Z}{\mathbb{Z}}
\newcommand{\expec}{\mathbb{E}}
\newcommand{\prob}{\mathbb{P}}

\newcommand{\rmd}{\mathrm{d}}

\newcommand{\qseq}{\mathbf{q}}

\newcommand{\map}{\mathfrak{m}}
\newcommand{\emap}{\mathfrak{e}}
\newcommand{\umap}{\mathfrak{u}}
\newcommand{\bmap}{\mathfrak{b}}

\newcommand{\rootedge}{e_{\mathrm{r}}}
\newcommand{\rootface}{f_{\mathrm{r}}}

\newcommand{\hopconst}{\mathcal{H}}

\newcommand{\maps}{\mathcal{M}}
\newcommand{\bmaps}[1]{\mathcal{M}^{(#1)}}
\newcommand{\edges}{\mathcal{E}}
\newcommand{\vertices}{\mathcal{V}}
\newcommand{\faces}{\mathcal{F}}

\newcommand{\lenconst}{\mathcal{L}}

\newcommand{\zp}{z^+}
\newcommand{\zd}{z^\diamond}

\newcommand{\compactfrac}[2]{{\textstyle\frac{#1}{#2}}}

\begin{document}

\title{The peeling process of infinite Boltzmann planar maps}

\author{Timothy Budd}
\dedicatory{
Niels Bohr Institute, University of Copenhagen\\
Blegdamsvej 17, 2100 Copenhagen \O , Denmark.\\
\vspace{3mm}
email: {\tt budd@nbi.dk}
}

\begin{abstract}
\noindent
We start by studying a peeling process on finite random planar maps with faces of arbitrary degrees determined by a general weight sequence, which satisfies an admissibility criterion.
The corresponding perimeter process is identified as a biased random walk, in terms of which the admissibility criterion has a very simple interpretation.
The finite random planar maps under consideration were recently proved to possess a well-defined local limit known as the infinite Boltzmann planar map (IBPM).
Inspired by recent work of Curien and Le Gall, we show that the peeling process on the IBPM can be obtained from the peeling process of finite random maps by conditioning the perimeter process to stay positive.
The simplicity of the resulting description of the peeling process allows us to obtain the scaling limit of the associated perimeter and volume process for arbitrary regular critical weight sequences.

\end{abstract}
\maketitle

\section{Introduction}

Discrete random surfaces encoded by maps have been subject of intensive research for a long time, in the mathematical literature mainly because of their rich combinatorial and geometric structure, and in the physics literature because of their relevance to two-dimensional quantum gravity, string theory, and Feynman diagrams.
Many exact enumerative properties of random (planar) maps, starting with the seminal work of Tutte in the sixties \cite{tutte_census_1963}, were known before methods were discovered to study their intrinsic geometric properties.
Perhaps the first significant progress in this direction was made in the early nineties, when Watabiki introduced a \emph{peeling procedure} of a random triangulated surface \cite{watabiki_construction_1995}.
This procedure shortly after lead to the first determination of a two-point function of large random triangulations \cite{ambjorn_scaling_1995}, which captures the probability distribution of the geodesic distance between two random vertices.\footnote{The derivation of the two-point function in \cite{ambjorn_scaling_1995} has always been regarded as heuristic, leading to an approximate probability distribution of the graph distance between random vertices that only becomes exact in the scaling limit. Only recently in \cite{ambjorn_multi-point_2014} it was recognized that the expressions in \cite{ambjorn_scaling_1995} are already exact in the discrete setting, but they do not correspond to the graph distance of the triangulation but to a first-passage time on the dual cubic map, which agree up to rescaling in the scaling limit.}
This two-point function was shown to possess a scaling limit if the distances were rescaled by the fourth root of the volume, leading to the (now well-established) conjecture that large random surfaces have a Hausdorff dimension of four.

In the late nineties Schaeffer \cite{schaeffer_conjugaison_1998}, inspired by earlier work by Cori and Vauquelin \cite{cori_planar_1981}, discovered a bijection between quadrangulations and labeled trees that was tailored to the study of graph distances.
Soon this bijection and its many generalizations, particularly the Bouttier-Di~Francesco-Guitter bijection \cite{bouttier_planar_2004}, became the tool of choice for studying geometry of random surfaces, culminating in the proof of the convergence in the Gromov-Hausdorff sense of many classes of large random planar maps to a universal random continuum geometry on the sphere, known as the \emph{Brownian map}.
Nevertheless, the peeling procedure in various forms received renewed interest in the mathematical literature, starting with its formalization in the setting of the uniform infinite planar triangulation (UIPT, \cite{angel_uniform_2003}) by Angel in  \cite{angel_growth_2002}.
It was recognized that peeling could not only be used to study geometry, but that due to its important Markov properties it could be tailored towards studying many other aspects of random maps, including various forms of percolation \cite{angel_growth_2002,angel_percolations_2013,menard_percolation_2014,ambjorn_multi-point_2014,richier_universal_2014,bjornberg_site_2015}, random walks \cite{benjamini_simple_2013}, and some aspects of their conformal structure \cite{curien_glimpse_2014}.

In a recent paper \cite{curien_scaling_2014} Curien and Le Gall revisited the geometric study of the peeling process on the UIPT (and also on the uniform infinite planar quadrangulation (UIPQ)), which was initiated by Angel in \cite{angel_growth_2002}.
They found that the perimeter process, which keeps track of the length of the (simple) boundary of the explored region during a peeling process of the UIPT, takes the form of a random walk conditioned to stay positive.
The unconditioned random walk is easily seen to scale to a index-3/2 stable process with only negative jumps.
The perimeter process can be obtained from this random walk by a simple Doob transform (or $h$-transform), and it follows from an invariance principle \cite{caravenna_invariance_2008} that it scales to the aforementioned stable process conditioned to stay positive.
The convergence is then extended to include the volume process, which keeps track of the volume of the explored region. 

The goal of this paper is to extend these results to the more general setting of $\qseq$-Boltzmann planar maps associated to arbitrary weight sequences $\qseq = (q_1,q_2,\ldots)$.
A (finite) $\qseq$-Boltzmann planar map is a random planar map with arbitrary face degrees, where the probability of picking a particular map is proportional to the product of over all faces $f$ of the weight $q_{\deg(f)}$ depending on the degree $\deg(f)$ of $f$.
Whenever this procedure leads to a well-defined probability distribution on (finite) planar maps, the sequence $\qseq$ is said to be \emph{admissible}. 
We will explicitly derive the law of the perimeter process for (pointed) $\qseq$-Boltzmann planar maps associated to such admissible weight sequences.
Like in the case of the UIPT, the perimeter process is realized as a Doob transform of a random walk which depends on $\qseq$.
Explicit conditions for a weight sequences to be admissible were worked out by Miermont in \cite{miermont_invariance_2006}, using the Bouttier-Di~Francesco-Guitter bijection \cite{bouttier_planar_2004}.
As we will see these quite non-trivial conditions can be recast into very simple conditions on the law of the random walk associated to the perimeter process.
In fact, in some sense the $\qseq$-Boltzmann planar map is more naturally characterized by this law than by the weight sequence itself.
In particular, the critical weight sequences $\qseq$, which are basically sequences that are on the border of being admissible and which are the appropriate ones to study large map limits, are related to random walks that do not drift.

It has recently been shown in \cite{stephenson_local_2014} for general critical weight sequence $\qseq$, and earlier in the more restricted setting of bipartite Boltzmann planar maps in \cite{bjornberg_recurrence_2014}, that $\qseq$-Boltzmann planar maps possess a well-defined local limit, known as the infinite $\qseq$-Boltzmann planar map ($\qseq$-IBPM), as the conditioned number of vertices is taken to infinity.
We will show that our peeling process naturally extends to the setting of the $\qseq$-IBPM, and that the associated perimeter process, like in the case of the UIPT, is obtained by conditioning the aforementioned random walk to stay positive.
Based on the supposed universality of the Brownian map, and its local limit as the Brownian plane, one expects the scaling limit of the perimeter process to agree with that of the UIPT.
Indeed, this will turn out to be the case for a large class of critical weight sequences $\qseq$, namely the \emph{regular critical} ones, and we will obtain explicit expressions for the scaling constants involved.

Interestingly, our description of the peeling process also covers $\qseq$-IBPMs associated to \emph{non-regular critical} weight sequences $\qseq$ which are heavy-tailed.
The related finite $\qseq$-Boltzmann planar maps were studied by Le Gall and Miermont in \cite{le_gall_scaling_2011}, and the metric spaces associated to the graph distances were shown to converge in the Gromov-Hausdorff sense to random continuum geometries distinct from the Brownian map.
These continuum geometries are characterized by holes corresponding to macroscopic faces in the planar maps and lead to a Hausdorff dimension strictly smaller than four.
We will not study the scaling limits of the non-regular $\qseq$-IBPMs in any detail, but we do notice that the associated perimeter processes seem to converge to stable processes (conditioned to stay positive) with an index smaller than $3/2$.
A natural next step, which we leave to future work, is to consider the peeling process in the setting of random planar maps coupled to $O(n)$ models which are known to be closely related to heavy-tailed $\qseq$-Boltzmann planar maps 
\cite{le_gall_scaling_2011,borot_recursive_2012,borot_more_2012}.

The simple characterization of the perimeter (and volume) process of the $\qseq$-IBPM opens up the possibility to study various aspects of its geometry, and percolation imposed on it.
Indeed, the results presented in this paper are independent of the chosen peel algorithm, meaning roughly that it does not matter in what direction or what order one chooses to explore the $\qseq$-IBPM using the peeling procedure.
By choosing appropriate peeling algorithms one can for instance discover the $\qseq$-IBPM by balls\footnote{Or, rather, the hulls of the balls, which are obtained from the balls by including the finite components of their complements, since during a peeling process the explored region is always kept simply-connected.} of increasing geodesic radius with respect to different metrics associated to the planar map, including graph distance, first-passage time and dual graph distance.
The known perimeter and volume processes then give information about the perimeter and volume of these geodesic balls.
Indeed, this was done for the UIPT in \cite{curien_scaling_2014} in the case of the graph distance and first-passage time (see also \cite{ambjorn_multi-point_2014} for the latter), and convergence was found to similar processes defined directly in terms of the Brownian plane \cite{curien_hull_2014}. 
We postpone similar applications of the peeling process in the setting of the $\qseq$-IBPM to a future paper, in which in particular we will use the technology in this paper to derive simple formulas for the relative scaling constants associated to the various metrics.

We should mention that the results in this paper should not be viewed as direct generalizations of the results by Curien and Le Gall.
First of all, the Boltzmann planar maps considered here will always be of \emph{type 1} meaning that they are allowed to contain multiple edges and loops.
More importantly, the \emph{lazy} peeling process considered here is slightly different from the (\emph{simple}) peeling process considered in \cite{curien_scaling_2014}.
The difference is that we do not require the \emph{frontier}, i.e. boundary of the explored region in the planar map, to be a simple closed curve, as it is in the simple peeling process.
We call this peeling process \emph{lazy} because the determination of the adjacency of a face incident to the frontier is postponed until it is needed in order to perform further exploration.
This means that it is possible, contrary to the simple peeling process, that an attempted peeling step does not lead to a newly explored face.
It turns out that this lazy peeling process is much simpler than the simple peeling process for general weight sequences, and corresponds in fact to a generalization of the original peeling process of Watabiki \cite{watabiki_construction_1995}.
As we will see the simplicity manifests itself in the existence of a martingale for the aforementioned random walk associated to the (lazy) perimeter process that is nearly universal, meaning that it depends on the weight sequence only through a single parameter $r\in(-1,1]$ (which moreover is always $r=1$ in the bipartite case).
It is exactly this martingale that can be used to bias the random walk to stay positive.

The paper is organized as follows. 
In section \ref{sec:mapsandpeeling} we introduce the lazy peeling process on finite pointed planar maps and on infinite maps, and in section \ref{sec:boltzmannmaps} we determine its law when applied to pointed $\qseq$-Boltzmann planar maps.
Some of the proofs depend explicitly on Miermont's conditions in \cite{miermont_invariance_2006} for the admissibility of weight sequences.
Since they involve different notation and appeal to the Bouttier-Di~Francesco-Guitter bijection \cite{bouttier_planar_2004}, which is not of central importance to the peeling process, we have decided to gather those proofs in appendix \ref{sec:app}.
In section \ref{sec:volume} we take a brief look at the volume process, while in section \ref{sec:ibpm} we derive and prove the law of the perimeter process of the $\qseq$-IBPM.
In sections \ref{sec:scaleperim} and \ref{sec:scalevol} we derive the scaling limits of the perimeter and volume processes on the $\qseq$-IBPM for regular critical $\qseq$, thereby reproducing the results of \cite{curien_scaling_2014}, Section 3, in the setting of the $\qseq$-IBPM.
We finish with some examples in Section \ref{sec:examples}, showing that the formulation in terms of random walks leads to an efficient way to calculate explicit scaling constants. 

\section{Planar maps and the lazy peeling process}\label{sec:mapsandpeeling}

A planar map $\map$ is a connected graph that is properly embedded in the sphere, i.e. without edge crossings, and two such maps are considered identical when they are related by a homeomorphism of the sphere. 
The sets of vertices, (unoriented) edges, and faces of $\map$ are denoted by $\vertices(\map)$, $\edges(\map)$, and $\faces(\map)$, respectively.
In addition we will denote by $\vec{\edges}(\map)$ the set of oriented edges of $\map$. 
Unless indicated otherwise we will always consider \emph{rooted} planar maps, which are planar maps with a distinguished oriented edge $\rootedge\in\vec{\edges}(\map)$, the \emph{root edge}.
The face on the left-hand side of the root edge $\rootedge$ is called the \emph{root face} $\rootface\in\faces(\map)$.
We denote the set of all (finite) rooted planar maps by $\maps$, while $\bmaps{k}$ denotes the rooted planar maps with root face degree $\deg(\rootface)=k$.
In addition we often consider \emph{pointed} planar maps $(\map,v)\in\maps_\bullet$, which are rooted planar maps with a distinguished vertex $v\in\vertices(\map)$.

Given a planar map $\map\in\maps$, let a \emph{closed path} $F = (e_i)_i$ be a finite sequence of oriented edges $e_i\in\vec{\edges}(\map)$ such that for each pair $(e_i,e_{i+1})$ of (cyclically) consecutive edges $e_{i+1}$ starts at the endpoint of $e_i$.
We say $F$ is \emph{non-intersecting} if $F$ contains no duplicates and whenever it visits a single vertex $v\in\vertices(\map)$ multiple times the sequence of edges of $e$ starting or ending at $v$ should be ordered counterclockwise around $v$ in the embedding.
A non-intersecting closed path $F$ naturally partitions the faces $\faces(\map)$ into those \emph{inside} and \emph{outside} $F$, where we say a face $f$ is \emph{inside} $F$ when $F$ winds around $f$ in counterclockwise direction, and \emph{outside} $F$ otherwise.

\begin{figure}[t]
\begin{center}
\includegraphics[width=.8\linewidth]{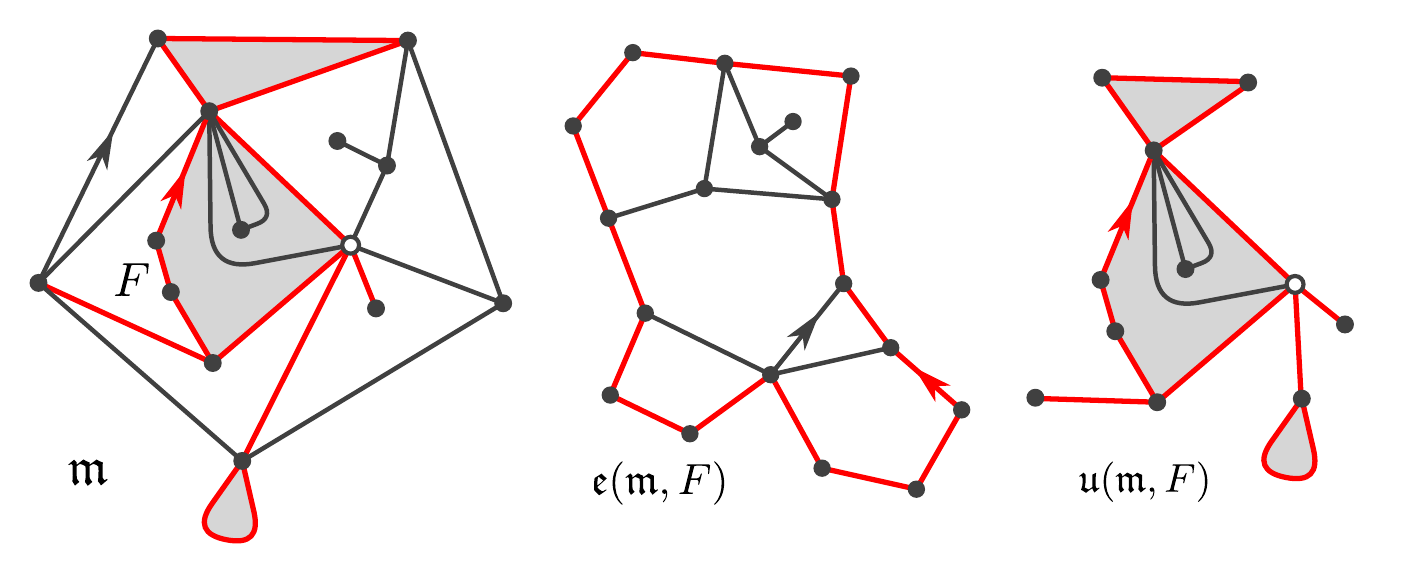}
\end{center}
\caption{An example of a pointed, rooted planar map $\map$ with a frontier $F$. The edges appearing in the frontier are colored red, and the first edge is indicated by an arrow. The second and third figure show the corresponding explored map $\emap(\map,F)$ and unexplored map $\umap(\map,F)$. } %
\label{fig:frontier}
\end{figure}

A non-intersecting path $F$ on a pointed planar map $(\map,v)\in\maps_\bullet$ is said to be a \emph{frontier} when the following conditions are satisfied: the root face $\rootface$ is inside $F$, the vertex $v$ is either on $F$ or incident to a face outside $F$, and any two faces inside $F$ are connected by a path in the dual map that does not cross $F$.
See figure \ref{fig:frontier} for an example.
The \emph{explored map} $\emap(\map,F)\in\maps$ is a rooted planar map with a distinguished oriented edge $e$ obtained from $\map$ by cutting it open along $F$ and replacing all faces outside $F$ by a single face, called the \emph{outer face}, and $e$ is taken to be the oriented edge corresponding to the first entry in the frontier.
Two faces inside $F$ that share an edge belonging to $F$ are not adjacent to each other along this edge in $\emap(\map,F)$.
Similarly, if the $F$ visits a vertex multiple times it gives rise to a corresponding number of distinct vertices on the boundary of the outer face.
This means that $\emap(\map,F)\in\maps$ is not quite a submap of $\map$ in the usual sense.
The \emph{unexplored map} $\umap(\map,F)\in\maps_\bullet$ is the submap of $\map$ obtained from $\map$ by removing all vertices and edges that are not on $F$ and are not incident to a face outside $F$.
We take $\umap(\map,F)$ to be rooted at the first oriented edge of the frontier.
Notice that both the outer face of the explored map $\emap(\map,F)$ and the root face of the unexplored map $\umap(\map,F)$ have degree equal to the length $|F|$ of the frontier.
The former is necessarily a simple face, i.e. all its corners correspond to distinct vertices, while the latter is in general not.
By convention we allow $F=\emptyset$, in which case we set $\emap(\map,\emptyset)=\map$ and $\umap(\map,\emptyset)=\{v\}$, i.e. the planar map consisting of just the vertex $v$.

\begin{figure}[t]
\begin{center}
\includegraphics[width=.8\linewidth]{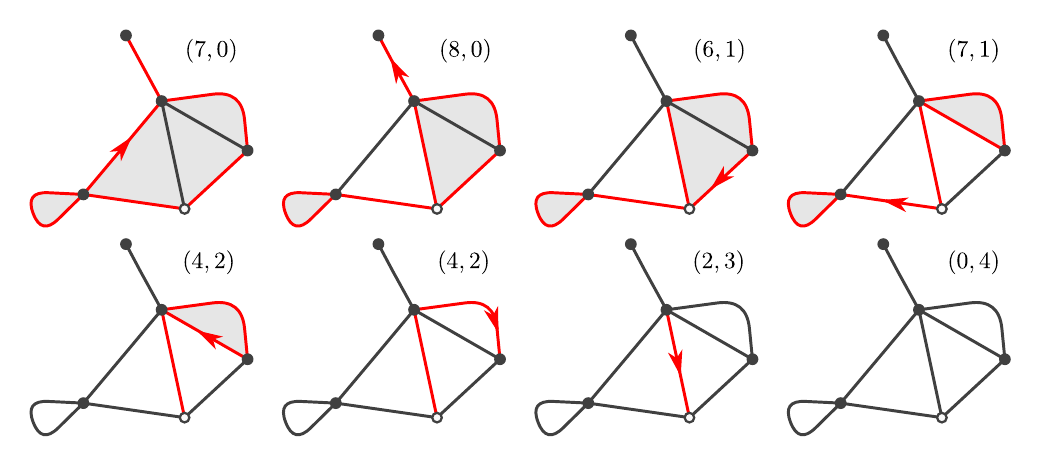}
\end{center}
\caption{An example of a lazy peeling process on a planar map with root face degree 7. The marked vertex is shown as an open dot, while the red edges represent the frontier. The peel edges, which are indicated by the red arrows, are chosen arbitrarily at each step. The numbers in parentheses correspond to the perimeter and volume $(l_i,V_i)_{i\geq 0}$.} %
\label{fig:diskpeeling}
\end{figure}

A \emph{lazy peeling process} on a pointed planar map $(\map,v)\in\maps_\bullet$ is a sequence $(F_i)_{i\geq0}$ of frontiers on $\map$ satisfying the following property.
For each $i\geq 0$, if $F_i=\emptyset$ then $F_{i+1}=\emptyset$, and otherwise the frontier $F_{i+1}$ is obtained from $F_{i}$ by \emph{peeling} an oriented edge $e\in F_i$, which we call the \emph{peel edge}.
We distinguish two types of peeling depending on whether the face on the right-hand side of $e$ is inside or outside $F_i$.
If it is inside, $e$ is incident on both sides to the root face in the unexplored map $\umap(\map,F_i)$ and therefore its removal from $\umap(\map,F_i)$ leads to two disconnected components, one of which contains the marked vertex $v$.
Then $F_{i+1}$ is the (possibly empty) subsequence of $F_i$ of edges that are contained in the latter component.
If the face $f$ on the right-hand side of $e$ is outside $F_i$, $F_{i+1}$ is obtained from $F_i$ by replacing its entry $e$ by the sequence of edges of $f$ in counterclockwise order starting at the starting point of $e$ and ending  at the endpoint of $e$.
We refer to the two types of peeling steps as \emph{pruning the frontier} and \emph{exploring a new face}, respectively.

It is easy to see that the number of edges $|\edges(\umap(\map,F_i))|$ in the unexplored map strictly decreases at each step as long as $F_i \neq \emptyset$.
Therefore, after a finite number of steps the frontier will vanish, meaning that the map $\map$ is fully explored.
Notice also that the lazy peeling process $(F_i)_{i\geq0}$ is completely fixed by the initial frontier $F_0$, which is usually taken to be the contour of the root face $\rootface$ starting at the root edge $\rootedge$, together with an algorithm to select a peel edge from the frontier $F_{i}$ at each step $i$.

Two integer-valued processes associated to the lazy peeling process $(F_i)_{i\geq 0}$ will play an important role in the rest of this paper.
The first is the \emph{perimeter process} $(l_i)_{i\geq0}$ which is simply defined as the length $l_i = |F_i|$ of the frontier.
The second is the \emph{volume process} $(V_i)_{i\geq 0}$ which counts the number of fully explored vertices after $i$ steps, i.e. $V_i = |\vertices(\emap(\map,F_i))| - |F_i|$ is the number of vertices in the explored map which are not incident to the outer face.
If $F_{i+1}$ is obtained from $F_i$ by exploring a new face of degree $k$, then $l_{i+1} = l_i + k - 2 \geq l_i-1$ and $V_{i+1}=V_i$.
If on the other hand $F_{i+1}$ is obtained by pruning $F_i$ then $l_{i+1} \leq l_i - 2$ and $V_{i+1} = V_i + v$ where $v\geq 1$ is the number of vertices in the ``discarded'' component of the unexplored map $\umap(\map,F_i)$ after removal of the peel edge.

The notion of lazy peeling can be easily extended to the case of an \emph{infinite rooted planar map} $\map_\infty\in\maps_{\infty}$.
We define the latter to be an embedding of an infinite, locally finite graph, i.e. a graph for which all vertices have finite degree, in $\R^2$ such that its complement is a disjoint union of simply connected, bounded domains.
Moreover we need to assume that $\map_\infty$ is \emph{one-ended}, meaning that the complement in $\map_\infty$ of each finite subset of vertices contains a unique infinite connected component, and that all faces have finite degree.

\begin{figure}[t]
\begin{center}
\includegraphics[width=.8\linewidth]{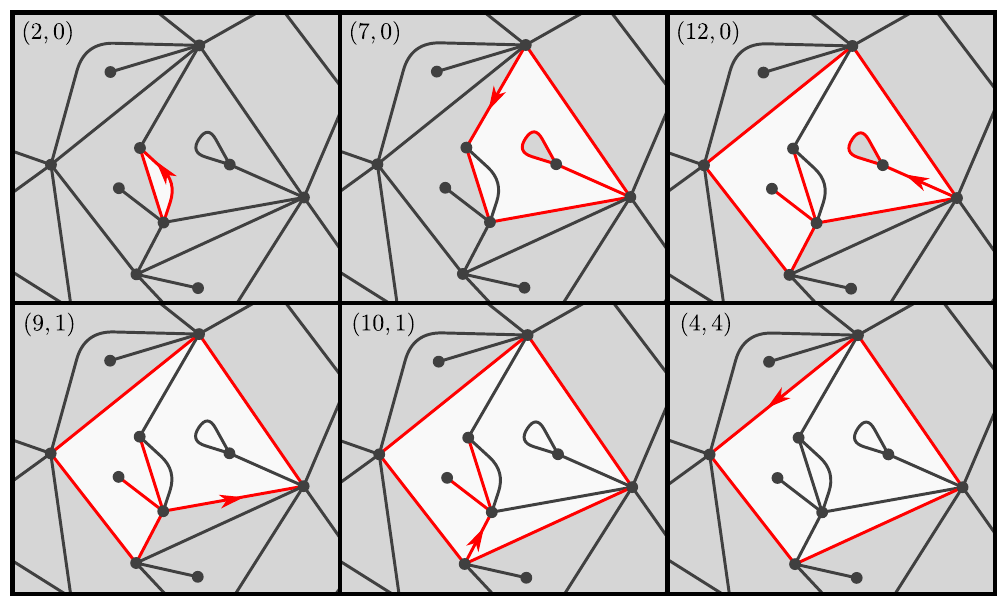}
\end{center}
\caption{A lazy peeling process of an infinite planar map. The numbers in parentheses again correspond to the perimeter and volume process $(l_i,V_i)_{i\geq 0}$. } %
\label{fig:infinitemap}
\end{figure}

We can define a frontier $F$ on $\map_\infty$ in the same way as before, except that we replace the condition that the marked vertex is outside $F$ by the condition that there are finitely many faces inside $F$.
The explored map $\emap(\map_\infty,F)$ is still a rooted planar map with a distinguished outer face of degree $|F|$, while the unexplored map $\umap(\map_\infty,F)$ is now an infinite rooted planar map with root face degree $|F|$.
The lazy peeling process is again defined as a sequence of frontiers $(F_i)_{i\geq 0}$ where the frontier $F_{i+1}$ is obtained from $F_i$ by peeling as before.
The only difference is that, when the frontier is pruned, the appropriate disconnected component of $\umap(\map,F_i)$ (after removal of the peel edge) is not selected to contain the marked vertex but to be infinite.
Notice that in particular the frontier can never become empty, i.e. $l_i = |F_i| \geq 1$ for all $i\geq 0$. 
See figure \ref{fig:infinitemap} for an example.

One should keep in mind that it is possible, with particular peeling algorithms, that a peeling process of an infinite planar map does not explore the full map.
For the random maps discussed below we do expect that the full map is almost surely explored independently of the peeling algorithm. 
However, proving this in general goes beyond the scope of this paper (see \cite{curien_scaling_2014}, Corollary 7, for a proof in the case of the ``simple'' peeling process on the UIPT).

\section{Boltzmann planar maps}\label{sec:boltzmannmaps}

We would like to apply the peeling process introduced in the previous section to a random pointed planar map $\map\in\maps_\bullet^{(l)}$ with fixed root face degree $l$.
A quite general probability distribution on $\maps_\bullet^{(l)}$ (and on $\maps^{(l)}$) is obtained by assigning a weight $w_\qseq(\map)$ to $\map$ given by
\begin{equation}\label{eq:weight}
w_\qseq(\map) := \prod_{f\in\faces(\map)\setminus\{\rootface\}} q_{\deg(f)},
\end{equation}
where the product is over all non-root faces and $\qseq=(q_1,q_2,\ldots)$ is a \emph{weight sequence} of non-negative real numbers.
To avoid degenerate situations we will require that at least one of the $q_k$'s with $k\geq 3$ is positive.
We define the \emph{disk function} $W^{(l)}(\qseq)$ and \emph{pointed disk function} $W_\bullet^{(l)}(\qseq)$ as
\begin{equation}\label{eq:diskfunctions}
W^{(l)}(\qseq) := \sum_{\map\in\bmaps{l}} w_\qseq(\map),\quad\quad W_\bullet^{(l)}(\qseq) := \sum_{\map\in\bmaps{l}_\bullet} w_\qseq(\map),
\end{equation}
where by convention we include the single vertex map in $\maps^{(0)}$ and $\maps_\bullet^{(0)}$ such that $W^{(0)}(\qseq) = W_\bullet^{(0)}(\qseq) = 1$.
Notice that in general $W^{(l)}(\qseq) \leq W_\bullet^{(l)}(\qseq)$. 

In order for the weights $w_\qseq$ to give rise to a probability distributions on $\maps^{(0)}$ and $\maps_\bullet^{(0)}$ we need $W_\bullet^{(l)}(\qseq)<\infty$, which puts a non-trivial restriction on the weight sequence $\qseq$.
If $q_k=0$ for all odd $k$ the pointed disk function $W_\bullet^{(l)}(\qseq)$ has support on bipartite planar maps and we call $\qseq$ \emph{bipartite}, and \emph{non-bipartite} otherwise.
It is not hard to see that the only way $W_\bullet^{(l)}(\qseq)$ can vanish is when $\qseq$ is bipartite and $l$ is odd.
As we will see below, apart from this degenerate situation where $W_\bullet^{(l)}(\qseq)=0$, the condition that $W_\bullet^{(l)}(\qseq) < \infty$ only depends on $\qseq$ and not on $l$.
Therefore, it suffices to consider the case $l=2$, and following \cite{miermont_invariance_2006} we will call $\qseq$ \emph{admissible} if $W_\bullet^{(2)}(\qseq) < \infty$.
To see why the case $l=2$ is special, notice that there is a natural bijection between pointed planar maps with root face degree $2$ (and more than one face) and pointed rooted planar map with arbitrary root face degree, corresponding to the operation of gluing the edges of the root face.
Hence we may identify
\begin{equation}\label{eq:partitionfunction}
W^{(2)}_\bullet(\qseq) - 1 = \sum_{\map\in\maps_\bullet}\prod_{f\in\faces(\map)}q_{\deg(f)}, 
\end{equation}
which is the \emph{partition function} studied in \cite{miermont_invariance_2006}.

In the following we will keep the dependence on $\qseq$ implicit whenever the considered weight sequence is clear from the context and simply write $W_\bullet^{(l)}=W_\bullet^{(l)}(\qseq)$.
We will also denote by $W(z)$ and $W_\bullet(z)$ the generating functions
\begin{equation}
W(z):= \sum_{l=0}^{\infty} W^{(l)} z^{-l-1}, \quad\quad W_\bullet(z):= \sum_{l=0}^{\infty} W_\bullet^{(l)} z^{-l-1},
\end{equation} 
which, as we will see below, converge for $|z|$ large enough.

\subsection{Lazy peeling of Boltzmann planar maps}
Let $(\map,v)\in\maps_\bullet^{(l)}$ be a pointed Boltzmann planar map and let $F_0$ be the frontier given by the contour of the root face, which has length $|F_0|=l$.
Recall that a lazy peeling process $(F_i)_{i\geq 0}$ is characterized by a \emph{peel algorithm} that selects a peel edge at each step $i$.
The algorithm can be deterministic or probabilistic, but in either case we restrict the choice of peel edge at step $i$, i.e. the step $F_i \to F_{i+1}$, only to depend on the explored map $\emap(\map,F_i)$.
In this case the lazy peeling satisfies the following important Markov property. 

\begin{proposition}
For any $i\geq 0$, conditional on the explored map $\emap(\map,F_i)$ after $i$ steps in the lazy peeling process with a peel algorithm as above, the law of the unexplored map $\umap(\map,F_i)\in\maps_\bullet^{(l_i)}$ depends only on the length $l_i=|F_i|$ of the frontier, and its distribution is given by 
\begin{equation}
\maps_\bullet^{(l_i)}\to\R : \map' \to \frac{w_\qseq(\map')}{W_\bullet^{(l_i)}}.
\end{equation}
\end{proposition}
\begin{proof}
For a fixed rooted planar map $\emap$ with a distinguished, simple outer face of degree $l$ the mapping $\umap : (\map,F) \to \umap(\map,F)$ determines a bijection
\begin{equation*}
\{(\map,F) : \map\in\maps_\bullet\text{, }F\text{ frontier, }\emap(\map,F) = \emap\} \to \maps_\bullet^{(l)}.
\end{equation*}
To see that this is a bijection, notice that its inverse is given by taking any pointed map $\map'\in \maps_\bullet^{(l)}$ and gluing its root face to the outer face of $\emap$ to obtain a pointed planar map with a frontier.
Notice that $w_\qseq(\map) = w_0 w_\qseq(\umap(\map,F))$, where $w_0$ is a constant independent of $\umap(\map,F)$.
Hence, if $\map$ is a $\qseq$-Boltzmann planar map and $F=F_i$ the frontier after $i$ steps in a peel algorithm, then the probability that $\umap(\map,F_i)=\map'$ for $\map'\in \maps_\bullet^{(l_i)}$ is $w_\qseq(\map')/W_\bullet^{(l_i)}$.
\end{proof}

In the light of this Markov property, it is often convenient to think of the peeling process not as a sequence of frontiers $(F_i)_{i\geq 0}$ on a random map $\map$, but as a sequence of explored maps $(\emap_i)_{i\geq 0}:=(\emap(\map,F_i))_{i\geq0}$.
We can easily determine the law of $\emap_{i+1}$ given $\emap_i$ and a peel edge $e$.
Recall that two types of peeling steps can occur: either a new face is explored or the frontier is pruned (see figure \ref{fig:peelstep}).
In the first case $\emap_{i+1}$ is obtained from $\emap_i$ by gluing a new face of degree $k$ to $e$, and $l_{i+1} = l_i + k - 2$.
It follows from the law of the unexplored map that this occurs with probability  
\begin{equation}\label{eq:posprob}
\prob(l_{i+1}-l_i=k-2 \geq -1 | \emap_i ) = q_k \frac{W_\bullet^{(l_i+k-2)}}{W_\bullet^{(l_i)}}.
\end{equation}
In the second case $\emap_{i+1}$ is obtained from $\emap_i$ by gluing $e$ to another edge $e'$ in the frontier, thereby splitting the outer face into two faces.
One of these two faces, say with degree $l$, is then \emph{filled in} with a (unpointed) Boltzmann planar map with the appropriate root face degree $l$, and $l_{i+1}=l_i - l - 2$.
By examining the Boltzmann weights involved one finds that this occurs with probability
\begin{equation}\label{eq:negprob}
\prob(l_{i+1}-l_i=-l-2\leq -2 | \emap_i ) = 2 W^{(l)}\frac{W_\bullet^{(l_i-l-2)}}{W_\bullet^{(l_i)}}.
\end{equation}

\begin{figure}[t]
\begin{center}
\includegraphics[width=.85\linewidth]{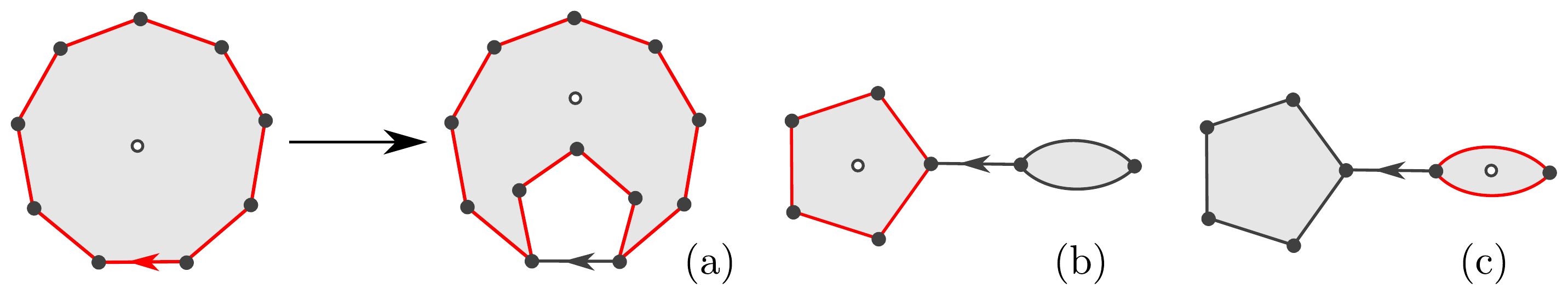}
\end{center}
\caption{Schematic depiction of three different outcomes of a (lazy) peeling step of a pointed planar map: either a new face is explored (a), or the frontier is pruned with the marked vertex to the left (b) or to the right (c) of the peel edge. The shaded faces with a dot represent the outer faces of the explored maps, while undotted shaded faces are to be filled in with (unpointed) Boltzmann planar maps with the appropriate root face degree. }%
\label{fig:peelstep}
\end{figure}

The fact that the probabilities (\ref{eq:posprob}) and (\ref{eq:negprob}) add up to one is equivalent to the \emph{loop equation}
\begin{equation}\label{eq:pointedloopeq}
W_{\bullet}^{(l)} = \sum_{k=1}^{\infty} q_k W_{\bullet}^{(l+k-2)}+2\sum_{l'=0}^{l-2}W^{(l')}W_{\bullet}^{(l-l'-2)},
\end{equation}
satisfied by the pointed disk function.
Similarly, by examining the face on the right-hand side of an unpointed planar map with root face degree $l$, one finds that the unpointed disk function satisfies a similar and well-known loop equation
\begin{equation}\label{eq:loopeq}
W^{(l)} =  \sum_{k=1}^{\infty} q_k W^{(l+k-2)} + \sum_{l'=0}^{l-2}W^{(l')}W^{(l-l'-2)}.
\end{equation}

It turns out that the pointed disk function $W_\bullet(z)$ has a universal form, that we will use as a starting point for our investigation.

\begin{proposition}\label{thm:pointeddisk}
Given an admissible weight sequence $\qseq=(q_1,q_2,\ldots)$, there exist real numbers $c_+>2$ and $-c_+\leq c_-<c_+$ such that the pointed disk function is finite for $z>c_+$ and is given by
\begin{equation}\label{eq:pointeddisk}
W_{\bullet}(z) = \frac{1}{\sqrt{(z-c_+)(z-c_-)}}.
\end{equation}
Moreover, $c_-=-c_+$ if and only if $\qseq$ is bipartite.
\end{proposition}

This result can in principle be derived directly from the loop equations (\ref{eq:pointedloopeq}) and (\ref{eq:loopeq}), using the so-called \emph{one-cut assumption}, see e.g. \cite{borot_more_2012} for a discussion in the case of general admissible weight sequences.
However, there is another particularly simple route towards (\ref{eq:pointeddisk}) using the Bouttier-Di~Francesco-Guitter bijection \cite{bouttier_planar_2004} between pointed planar maps and labeled mobiles.
The identity (\ref{eq:pointeddisk}) is well-known, appearing for instance in a slightly different form in \cite{bouttier_planar_2011}.
For completeness we include a proof in the appendix, Section \ref{sec:proofpointeddisk}, using very similar arguments as the ones in \cite{bouttier_planar_2004,bouttier_planar_2011}.
Notice in particular that Proposition \ref{thm:pointeddisk} implies that $W_\bullet^{(l)}<\infty$ for all $l\geq 0$ when $\qseq$ is admissible, as promised.   

The form of the unpointed disk function $W(z)$ is not universal, but can be easily derived from the loop equation (\ref{eq:pointedloopeq}).
To see this let us introduce the generating functions
\begin{align}
U'(z) &:= \sum_{k=1}^{\infty} q_k z^{k-1} \label{eq:potential} \\
M(z) &:= -1+\sum_{k=2}^{\infty}\sum_{l=0}^{k-2}q_k  W_{\bullet}^{(k-l-2)}z^{l}.\label{eq:mfunction}
\end{align}
From Proposition \ref{thm:pointeddisk} it follows that asymptotically 
\begin{equation}\label{eq:wdotasymp}
\lim_{l\to\infty} W_{\bullet}^{(l+k)}/W_{\bullet}^{(l)} = c_+^k.
\end{equation}
Therefore (\ref{eq:pointedloopeq}) implies that
\begin{equation*}
U'(c_+) = \sum_{k=1}^{\infty} q_k \lim_{l\to\infty} \frac{W_{\bullet}^{(l+k-2)}}{W_{\bullet}^{(l)}} \leq \lim_{l\to\infty} \sum_{k=1}^{\infty} q_k \frac{W_{\bullet}^{(l+k-2)}}{W_{\bullet}^{(l)}} \leq 1,
\end{equation*}
meaning that $U'(z)$ has a radius of convergence $R\geq c_+$.
In terms of generating functions the loop equation (\ref{eq:pointedloopeq}) reads
\begin{equation*}
(z-U'(z)-2W(z))W_{\bullet}(z)+M(z)=0\quad(c_+\leq |z|\leq R), 
\end{equation*}
which should be understood as an equation for Laurent series on the annulus $c_+\leq |z|\leq R$.
In particular this implies that $M(z)$ also has a radius of convergence at least equal to $c_+$.
In combination with Proposition \ref{thm:pointeddisk} this leads to the familiar \emph{one-cut} form of the disk function,
\begin{align}
W(z) &= \frac{1}{2}\left(z-U'(z)+M(z)W_{\bullet}(z)^{-1}\right) \nonumber\\
&= \frac{1}{2}\left(z-U'(z)+M(z)\sqrt{(z-c_+)(z-c_-)}\right),\label{eq:usualdiskfun}
\end{align}
which is finite for $|z| \geq c_+$.
One can easily check that (\ref{eq:usualdiskfun}) is indeed (a generating function for) a solution to the loop equation (\ref{eq:loopeq}).

\subsection{A new admissibility criterion}

It follows from the law of the peeling process, characterized by the probabilities (\ref{eq:posprob}) and (\ref{eq:negprob}), that the perimeter process $(l_i)_{i\geq 0}$ is a Markov process with step probabilities
\begin{equation}\label{eq:lengthprob}
\prob(l_{i+1}=l+k|l_i=l) = 
 \begin{cases}
   q_{k+2} \frac{W_{\bullet}^{(l+k)}}{W_{\bullet}^{(l)}} & \text{for } k \geq -1   \\
   2W^{(-k-2)} \frac{W_{\bullet}^{(l+k)}}{W_{\bullet}^{(l)}} & \text{for } -l \leq k < -1\\
   0 & \text{for } k < -l
 \end{cases},
\end{equation}
for $l>0$, while we use the convention $\prob(l_{i+1}=k|l_i=0)=\delta_{k,0}$.
The asymptotics (\ref{eq:wdotasymp}) imply that the large-$l$ limit of (\ref{eq:lengthprob}) corresponds to a random walk $(X_i)_{i\geq0}$ with step probabilities $\nu(k)$ given by
\begin{equation}\label{eq:nudef}
\nu(k) := \lim_{l\to\infty} \prob(l_{i+1}=l+k|l_i=l) = \begin{cases}
   q_{k+2} c_+^{k} & \text{if } k \geq -1 \\
   2W^{(-k-2)} c_+^{k} & \text{if } k \leq -2
 \end{cases}.
\end{equation}
It follows from Proposition \ref{thm:pointeddisk} that $c_+^{-l} W_\bullet^{(l)}$ only depends on $l$ and the ratio $r:=-c_-/c_+ \in(-1,1]$ and is given by
\begin{equation}\label{eq:h0fundef}
h_r^{(0)}(l) := [y^{-l-1}]\frac{1}{\sqrt{y-1}\sqrt{y+r}}=\left(\frac{-r}{4}\right)^{l} \binom{2l}{l}\, _2F_1\left(\frac{1}{2},-l;\frac{1}{2}-l;-\frac{1}{r}\right),
\end{equation}
where $_2F_1$ is the hypergeometric function defined as $_2F_1(a,b;c;z)=\sum_{n=0}^\infty \frac{(a)_n(b)_n}{(c)_n}\frac{z^n}{n!}$ in terms of the rising Pochhammer symbol $(a)_n := a (a+1)\ldots (a+n-1)$.
In general the values $h_r^{(0)}(l)$ for fixed $l\geq 0$ are polynomials in $r$ of order $l$, with the first few reading $h_r^{(0)}(0) = 1$, $h_r^{(0)}(1) = (1-r)/2$, and $h_r^{(0)}(2) = (3-2r+3r^2)/8$. 
For convenience we set $h_r^{(0)}(l)=0$ for $l<0$.

In terms of $\nu$ and $h_r^{(0)}$ the loop equation (\ref{eq:pointedloopeq}) takes on the simple form
\begin{equation}\label{eq:harmonicprime}
\sum_{l=-\infty}^{\infty}h_r^{(0)}(l+k)\nu(l) = h_r^{(0)}(k)\quad \text{for}\quad k>0.
\end{equation}
Whenever this relation holds we will say that $h_r^{(0)}$ is \emph{$\nu$-harmonic} on the positive integers.

Clearly the map $\qseq\to \nu$, which associates a law for a random walk to an admissible weight sequence, is injective since the weights may be recovered through
\begin{equation}\label{eq:weightfromnu}
q_k = \left(\frac{\nu(-2)}{2}\right)^{\frac{k-2}{2}} \nu(k-2).
\end{equation}
The following result is proved in the appendix, Section \ref{sec:proofwalksequence}, and relies on an explicit evaluation of the conditions for admissibility in \cite{miermont_invariance_2006}.

\begin{proposition}\label{thm:walksequence}
Relation (\ref{eq:weightfromnu}) determines a bijection between admissible weight sequences $\qseq$ and random walks $(X_i)_{i\geq 0}$ with step probabilities $\nu$ for which there exists an $r\in(-1,1]$ such that $h_r^{(0)}$ is $\nu$-harmonic on $\Z_{>0}$ and
\begin{equation}\label{eq:admissiblecondition}
\sum_{l=0}^{\infty} h_r^{(1)}(l+1)\nu(l)  \leq 1,
\end{equation}
with $h_r^{(1)}(l):=\sum_{p=0}^{l-1}h_r^{(0)}(l) = [y^{-l-1}]\frac{1}{(y-1)^{3/2}\sqrt{y+r}}$.
\end{proposition}

\begin{remark}
Actually we will see below in Proposition \ref{thm:drift} that the condition (\ref{eq:admissiblecondition}) is redundant, in the sense that it is satisfied by all $\nu$ for which $h_r^{(0)}$ is $\nu$-harmonic on $\Z_{>0}$ for some $r\in(-1,1]$. 
\end{remark}

Before continuing let us extend the definitions of $h_r^{(0)}$ and $h_r^{(1)}$ by introducing the family of functions $h_r^{(k)}(l):\Z\to\R$, $k\in\Z$, that will appear in various places below.
For $-1 \leq r \leq 1$ and $k\in\Z$ we set 
\begin{align}
h_r^{(k)}(l) &:= [y^{-l-1}]\frac{1}{(y-1)^{k+1/2}\sqrt{y+r}}\label{eq:hfundef}\\
&=\left(\frac{-r}{4}\right)^{l-k} \binom{2l-2k}{l-k}\, _2F_1\left(\frac{1}{2}+k,k-l;\frac{1}{2}+k-l;-\frac{1}{r}\right).\nonumber
\end{align}
Then $h_r^{(k)}(l) = 0$ for $l<k$, $h_r^{(k)}(k)=1$ and
\begin{equation}\label{eq:hfunrelations}
h_r^{(k+1)}(l) = \sum_{p=k}^{l-1}h_r^{(k)}(p), \quad h_r^{(k)}(l) = h_r^{(k+1)}(l+1) - h_r^{(k+1)}(l)\quad\text{for } k\geq 0.
\end{equation}
For future use we also record the asymptotics
\begin{equation}\label{eq:hfunasymp}
h_r^{(k)}(l)\sim \frac{l^{k-1/2}}{\Gamma(k+1/2)\sqrt{1+r}}\quad \text{as }l\to\infty\quad(r<1),
\end{equation}
which is valid for $r\in(-1,1)$ and $k\in\Z$.
In the case $r=1$, one has to be a bit careful since the asymptotics depend on the parity of $l$, but (\ref{eq:hfunasymp}) remains valid when the left-hand side is replaced by the average $(h_r^{(k)}(l)+h_r^{(k)}(l+1))/2$.
Finally, let us note that in the bipartite case, $h_1^{(0)}$ and $h_1^{(1)}$ have the simple expressions
\begin{equation}
h_1^{(0)}(2l) = 2^{-2l}\binom{2l}{l}, \quad h_1^{(1)}(2l) = 2l \cdot 2^{-2l}\binom{2l}{l},\quad(l\geq 0)
\end{equation}
while $h_1^{(0)}$ vanishes on the odd integers and $h_1^{(1)}(2l-1)=h^{(1)}_1(2l)$.

We call $\nu$ \emph{critical} when equality holds in (\ref{eq:admissiblecondition}) and \emph{regular critical} if in addition $\nu(k)$ falls off at least exponentially as $k\to\infty$, i.e., there exists a constant $C>1$ such that $\sum_{k=0}^{\infty} \nu(k) C^k < \infty$.
It is a simple check using the ingredients of the proof of Proposition \ref{thm:walksequence} that these are equivalent to the corresponding conditions on the weight sequence $\qseq$ as given in \cite{miermont_invariance_2006}.

Let us try to understand the condition (\ref{eq:admissiblecondition}) a bit better.
Notice that if $h_r^{(0)}$ is $\nu$-harmonic on $\Z_{>0}$, it follows from (\ref{eq:hfunrelations}) that the condition (\ref{eq:admissiblecondition}) is equivalent to
\begin{equation}\label{eq:admissiblecondition2}
\sum_{l=-\infty}^{\infty}h_r^{(1)}(l+k)\nu(l) - h_r^{(1)}(k) = \sum_{l=0}^{\infty}h_r^{(1)}(l+1)\nu(l) - h_r^{(1)}(1) \leq 0
\end{equation}
for all $k\geq 1$, i.e., $h_r^{(1)}$ is \emph{$\nu$-superharmonic} on $\Z_{>0}$.
Moreover, $\qseq$ is critical if and only if $h_r^{(1)}$ is $\nu$-harmonic on $\Z_{>0}$, since the latter implies the $\nu$-harmonicity of $h_r^{(0)}$. 
We therefore have

\begin{corollary}\label{thm:criticalwalk}
Relation (\ref{eq:weightfromnu}) determines a bijection between critical weight sequences $\qseq$ and step probabilities $\nu$ for which $h_r^{(1)}$ is $\nu$-harmonic on $\Z_{>0}$ for some $r\in(-1,1]$.
\end{corollary}

\subsection{The perimeter process as a Doob transform}

As we will see now the functions $h_r^{(0)}$ and $h_r^{(1)}$ have simple interpretations in terms of the random walk $(X_i)_{i\geq 0}$.

\begin{lemma}\label{thm:hitzero}
If $h_r^{(0)}$ is $\nu$-harmonic on $\Z_{>0}$ then the random walk $(X_i)_i$ started at $l$ hits zero before hitting the negative integers with probability $h_r^{(0)}(l)$.
\end{lemma}
\begin{proof}
For fixed $k>X_0=l$, let $\tau_k$ be the first time at which $(X_i)_i$ exits $\{1,2,\ldots k-1\}$. 
The sought-after probability can then be expressed as $\sup_k \prob(X_{\tau_k}=0)$.
Since $h_r^{(0)}$ is $\nu$-harmonic on $\Z_{>0}$, $h_r^{(0)}(X_i)$ defines a bounded martingale with respect to the random walk killed upon entry of $\Z_{\leq 0}$. 
By the optional stopping theorem 
\begin{equation*}
h_r^{(0)}(l) = \prob(X_{\tau_k}=0) + \sum_{p=k}^\infty \prob( X_{\tau_k} = p ) h_r^{(0)}(p).
\end{equation*}
The sum on the right-hand side is bounded by $h_r^{(0)}(k) + h_r^{(0)}(k+1)$ since $h_r^{(0)}$ is decreasing along the even and odd integers, and therefore it goes to zero as $k\to\infty$.
It follows that $\sup_k \prob(X_{\tau_k}=0) = h_r^{(0)}(l)$.
\end{proof}

Since the probability $h_r^{(0)}(l)$ is positive we can easily condition $(X_i)_i$ to hit zero before hitting the negative integers.
The result is again a Markov process, known as the \emph{Doob transform} (or $h$-transform) of $(X_i)_i$ with respect to $h_r^{(0)}$, which turns out to be identical in law to the perimeter process $(l_i)_i$.
Indeed, (\ref{eq:lengthprob}) can be rewritten precisely as
\begin{equation}\label{eq:htrans}
\prob(l_{i+1}=l+k|l_i=l) = \frac{h_r^{(0)}(l+k)}{h_r^{(0)}(l)}\nu(k),
\end{equation}
where $\nu$ is the law of the increments of $(X_i)_i$.

Following \cite{bertoin_conditioning_1994} we say a random walk $(X_i)_{i\geq 0}$ \emph{drifts to $\infty$} if $\prob(X_i \geq X_0\text{ for all }i) > 0$, and it \emph{drifts to $-\infty$} if $\prob(X_i \leq X_0\text{ for all }i) > 0$.
If it drifts neither to $\infty$ nor to $-\infty$, it is said to \emph{oscillate}.
Using that $h_r^{(1)}$ is essentially the \emph{renewal function} associated to $(X_i)_{i\geq 0}$, we can apply the ingredients of \cite{bertoin_conditioning_1994} to show the following.

\begin{proposition}\label{thm:drift}
If $h_r^{(0)}$ is $\nu$-harmonic on $\Z_{>0}$ then the condition (\ref{eq:admissiblecondition}) is always satisfied.
Furthermore, if $\qseq$ is admissible, the associated random walk $(X_i)_i$ does not drift to $\infty$ and it oscillates if and only if $\qseq$ is critical.
%
%
\end{proposition}
\begin{proof}
Following \cite{bertoin_conditioning_1994} let $(H_i,T_i)_{i\geq 0}$ be the \emph{strict ascending ladder point process} of $(X_0-X_i)_{i\geq 0}$ defined by $T_0=0$ and 
\begin{equation*}
H_i = X_0-X_{T_i},\quad T_{i+1}=\inf\{j>T_i : X_0-X_j > H_i\},
\end{equation*}
and we set $H_i=\infty$ if $T_i = \infty$.
Then the \emph{renewal function} is given by $V(k) := \sum_{i=0}^\infty \prob(H_i \leq k)$.
Using Lemma \ref{thm:hitzero} we may identify $h_r^{(0)}(l) = \sum_{i=0}^{\infty} \prob(H_i = l)$.
Therefore $V(k) = \sum_{l=0}^{k} h_r^{(0)}(l) = h_r^{(1)}(k+1)$.

According to \cite{bertoin_conditioning_1994}, Section 2, the renewal function of any random walk is $\nu$-superharmonic on the non-negative integers.
Hence, $h_r^{(1)}$ is $\nu$-superharmonic on the positive integers, which as discussed before is equivalent to (\ref{eq:admissiblecondition}) when $h_r^{(0)}$ is $\nu$-harmonic on $\Z_{>0}$.
The second statement follows from the fact, see again \cite{bertoin_conditioning_1994}, Section 2, that a random walk drifting to $\infty$ has a bounded renewal function, while $h_r^{(1)}$ is unbounded.
Moreover, the renewal function of a random walk is $\nu$-harmonic on $\Z_{\geq 0}$ if and only if it does not drift to $-\infty$, which combined with Corollary \ref{thm:criticalwalk} gives the final statement.
\end{proof}

In the following we will consider cases when $\qseq$ is critical (but not necessarily regular critical). 
Since $h_r^{(1)}$ is $\nu$-harmonic on $\Z_{>0}$ and vanishes on $\Z_{\leq 0}$ the Doob transform of the random walk $(X_k)_{k\geq 0}$ w.r.t. $h_r^{(1)}$ corresponds exactly to conditioning $X_k$ to stay positive (see \cite{bertoin_conditioning_1994} for precise statements).
It is natural to expect that this transform describes the perimeter process of the infinite Boltzmann planar map.
In section \ref{sec:ibpm} we will see that this is indeed the case.

\section{The volume process}\label{sec:volume}

Recall that the volume process $(V_i)_{i\geq0}$ associated to a planar map $\map$ counts the number of fully explored vertices after $i$ peeling steps, i.e. $V_i = |\vertices(\emap(\map,F_i))| - |F_i|$, and that it only increases when the frontier is pruned.
If $\map$ is a $\qseq$-Boltzmann planar map and $l\geq0$, then conditional on the change $l_{i+1}-l_i = -l-2$ in the perimeter during the $i$'th peeling step, the law of the change $V_{i+1}-V_i$ in the volume is simply given by the law of the number of vertices $|\vertices(\map')|$ in a (unpointed) $\qseq$-Boltzmann planar map $\map'\in\maps^{(l)}$ with root face degree $l$.
Remember that when $l=0$ the latter is the one-vertex map and therefore $V_{i+1}-V_i=1$.
In order to study this law let us introduce the analogues $W^{(l,V)}$ and $W_{\bullet}^{(l,V)}$ of the unpointed and pointed disk functions (\ref{eq:diskfunctions}) where the sum is restricted to be over maps with exactly $V$ vertices, with the convention that $W^{(0,V)} = W_{\bullet}^{(0,V)} = \delta_{V,1}$.
Then
\begin{equation*}
\prob( V_{i+1}=V_i+v | l_{i+1}=l_i+l ) = \frac{W^{(l,v)}}{W^{(l)}}.
\end{equation*}

We denote the corresponding generating functions by
\begin{equation}\label{eq:diskfunctionvertexweight}
W^{(l)}_g := \sum_{V=1}^{\infty} W^{(l,V)} g^V = \sum_{\map\in\bmaps{l}} g^{|\vertices(\map)|} \prod_{f\in\faces(\map)\setminus\{\rootface\}} q_{\deg(f)}
\end{equation}
and similarly for the pointed version, which necessarily have radius of convergence in $g$ larger or equal to one.
Using Euler's formula we find that
\begin{equation*}
|\vertices(\map)| = 1 + l/2 + \sum_{f\in\faces(\map)\setminus\{\rootface\}} (\deg(f)-2)/2,
\end{equation*}
and therefore we may identify $W^{(l)}_g(\qseq) = g^{1+l/2} W^{(l)}(\qseq_g)$ where $\qseq_g$ is the weight sequence determined by $(q_g)_{k} := g^{(k-2)/2}q_k$.
The same relation holds between the pointed versions of the disk function, i.e. $W^{(l)}_{\bullet,g}(\qseq) = g^{1+l/2} W^{(l)}_\bullet(\qseq_g)$.

When $\qseq$ is admissible and $0<g<1$, then $\qseq_g$ is also admissible, and we denote by $c_{\pm}(g)$ the associated constants from Proposition \ref{thm:pointeddisk}, and $r(g) = -c_-(g)/c_+(g)$.
No general explicit expression is known for these functions, but it is proved in the appendix, Section \ref{sec:proofvolumeexpansion}, that the behavior in the vicinity of $g=1$ is universal in the case of regular critical weight sequences $\qseq$.
A similar derivation appears in \cite{ambjorn_quantum_1997}, section 4.2.5.

\begin{lemma}\label{thm:volumeexpansion}
If $\qseq$ is regular critical, $0<g<1$ and $c_\pm(g)$ is defined as above then
\begin{equation}
\lim_{g\uparrow 1} \frac{1-c_+(g)/c_+(1)}{\sqrt{1-g}} = \sqrt{\frac{16}{3(1+r)c_+^2\lenconst_\nu}},
\end{equation}
where
\begin{equation}\label{eq:lenconst}
\lenconst_\nu := \sum_{k=1}^\infty \nu(k) h_r^{(2)}(k+1).
\end{equation}
If $\qseq$ is non-bipartite we furthermore have that $\lim_{g\uparrow 1} (1-c_-(g)/c_-(1))/\sqrt{1-g} = 0$.
\end{lemma} 

For general critical weight sequences $\qseq$ it is not straightforward\footnote{See e.g. \cite{borot_recursive_2012}, Section 3.4, of \cite{borot_more_2012}, Section 6, for discussions of the asymptotic behavior in the non-regular case that $q_k\sim k^{\alpha-1}c_+^{-k}$ for $1/2 \leq \alpha \leq 3/2$.} to determine the behavior as $g\to 1$, but at least we have the following result that we will need in the next section. See section \ref{sec:proofrovdisk} for the proof.

\begin{lemma}\label{thm:rovdisk}
If $\qseq$ is critical, $l\geq1$ and $W_g^{(l)}(\qseq)$ is not identically zero, then $W_g^{(l)}(\qseq)$ has unit radius of convergence.
\end{lemma}

\section{The infinite $\qseq$-Boltzmann planar map}\label{sec:ibpm}

Depending on the admissible weight sequence $\qseq$, the $\qseq$-Boltzmann planar map $\map\in\maps^{(l)}$ with root face degree $l$ may have combinatorial restrictions on the number of its vertices $|\vertices(\map)|$.
The set of possible vertex numbers is denoted by $I_{\qseq,l}$, i.e.,
\begin{equation*}
I_{\qseq,l} := \{V\geq 1 : \prob(|\vertices(\map)|=V)>0\}.
\end{equation*}
In general, it follows from Euler's formula that $I_{\qseq,l}$ consists of all but finite of the integers in the lattice $c + d\, \Z_{\geq 0}$ for some $c\in\Z$ and
\begin{equation*}
d:= \gcd(\{k\geq 1 : q_{2k+2}>0\text{, or }k\text{ odd and }q_{k+2}>0\}).
\end{equation*}

The starting point of this section is the following recent result.
\begin{theorem}[Local limit of $\qseq$-Boltzmann planar maps, Stephenson \cite{stephenson_local_2014}]\label{thm:stephenson}
Let $\qseq$ be a critical weight sequence, and let $\map_V$, $V\in I_{\qseq,2}$, be rooted pointed $\qseq$-Boltzmann planar maps conditioned to have $V$ vertices.
Then there exists a random rooted infinite planar map $\map_\infty$ such that $\map_V$ converges in distribution in the local topology to $\map_\infty$ as $V\to\infty$ along $I_{\qseq,2}$. 
\end{theorem}

We would like to study lazy peeling processes on the infinite planar map $\map_\infty$.
Recall from Section \ref{sec:mapsandpeeling} that in order for such process to be well-defined one needs all faces to have finite degree and $\map_\infty$ to be one-ended.
Luckily both are satisfied by the $\qseq$-IBPM.
The proof of Theorem \ref{thm:stephenson} in \cite{stephenson_local_2014} relies on the fact that the $\qseq$-IBPM can be coded by an infinite multi-type Galton-Watson tree.
The absence of infinite degree faces in $\map_\infty$ follows from the absence of infinite degree vertices in the latter tree.
On the other hand, the one-endedness, which states that the complement of any finite submap has exactly one infinite component, follows from the fact that the tree has a unique infinite spine.

Similar convergence results were obtained in \cite{stephenson_local_2014} for the case where the conditions are put on the number of edges or faces instead of vertices, but only in the case of regular critical weight sequences.
In order to study the peeling process of the general infinite $\qseq$-Boltzmann planar maps, we therefore need to understand the peeling of $\qseq$-Boltzmann planar map $\map_V$ with fixed number $V$ of vertices.


\begin{theorem}\label{thm:peelingibpm}
The perimeter process $(l_i)_{i\geq 0}$ associated to any lazy peeling process of the infinite $\qseq$-Boltzmann planar map $\map_\infty$ determined by the critical weight sequence $\qseq$ is given by the Doob transform of the associated random walk from Corollary \ref{thm:criticalwalk} with respect to $h_r^{(1)}$.
\end{theorem}
\begin{proof}
Let $B_r(\map)$ be the geodesic ball of radius $r$ around the root in the pointed Boltzmann planar map $\map$, i.e. the submap of $\map$ consisting of all vertices, and edges between them, whose graph distance to the starting point of the root edge is at most $r$.
The $\qseq$-IBPM $\map_\infty$, whose existence follows from Theorem \ref{thm:stephenson}, is characterized by the fact that its geodesic balls $B_r(\map_\infty)$ agree in law with the balls $B_r(\map)$ of the map $\map$ when conditioned to have a large number of vertices, i.e.,
\begin{equation*}
\prob( B_r(\map_\infty) = \bmap ) = \lim_{V\to\infty} \prob( B_r(\map) = \bmap |\, |\vertices(\map)|=V).
\end{equation*}
The existence of these limits imply that we may instead condition the number of vertices in $\map$ to be at least $V$ since
\begin{equation*}
\lim_{V\to\infty} \prob( B_r(\map) = \bmap |\, |\vertices(\map)|\geq V) = \lim_{V\to\infty} \prob( B_r(\map) = \bmap |\, |\vertices(\map)|= V).
\end{equation*}
In the following we will denote the random map $\map$ conditioned to have at least $V$ vertices by $\map_{\geq V}$.
Similarly, we will write $W_{\bullet}^{(l,\geq V)} := \sum_{v=V}^{\infty}W_{\bullet}^{(l,v)}$ for the disk function with a lower bound $V$ on the number of vertices.

We will attempt to determine the law of the perimeter $(l_i(\map_\infty))_{i=0}^{n}$ by considering the $V\to\infty$ limit of the analogous process $(l_i(\map_{\geq V}))_{i=0}^{n}$ on $\map_{\geq V}$.
Unfortunately, due to the conditioning the latter is not a Markov process. 
However, we can easily turn it into a Markov process $(l_i,V_i)_{i=0}^{n}$ by considering the joint law of the perimeter $l_i$ and the volume $V_i$, i.e., the number of \emph{fully explored vertices} after $i$ peeling steps, see section \ref{sec:mapsandpeeling}.
Indeed, it is easy to see that the law of the unexplored map $\umap(\map_{\geq V},F_i)$ only depends on the length of the frontier and the number of fully explored vertices in the explored map $\emap(\map_{\geq V},F_i)$.
The transition probabilities follow from the loop equation
\begin{equation}\label{eq:wbulletV}
W_{\bullet}^{(l,\geq V)} =  \sum_{k=1}^{\infty} q_k W_{\bullet}^{(l+k-2,\geq V)}+2\sum_{l'=0}^{l-2}\sum_{v=1}^{\infty}W^{(l',v)}W_{\bullet}^{(l-l'-2,\geq V-v)},
\end{equation}
which is the straightforward generalization of (\ref{eq:pointedloopeq}).
We read off that
\begin{equation}\label{eq:lenvolprob}
\prob_{\map_{\geq V}}( l_{i+1} = l_i + k, V_{i+1} = V_i+v | \emap  ) =  \frac{W_{\bullet}^{(l_i+k,\geq V-V_i+v)}}{W_{\bullet}^{(l_i,\geq V-V_i)}}\begin{cases}
\delta_{v,0} q_{k+2} & k\geq -1\\
2 W_{v}^{(-k-2)}& -l_i\leq k\leq-2\\
0 & \text{otherwise.}
\end{cases}
\end{equation}

The limit as $V\to\infty$ of (\ref{eq:lenvolprob}) with $l_i,V_i,k,v$ fixed must exist and must equal the transition probability of the corresponding process on $\map_\infty$, since the interior of any frontier of length $l_i+k$ enclosing $V_i+v$ vertices is necessarily contained in a sufficiently large ball $B_{r}(\map_\infty)$ (take e.g. $r=V_i+l_i+v+k$).
In particular, setting $k=-2$, we find that the limit 
\begin{equation}\label{eq:wdotmin2}
\lim_{V\to\infty} W_\bullet^{(l-2,\geq V-1)}/W_{\bullet}^{(l,\geq V)} <\infty 
\end{equation}
exists for all $l\geq 3$. 
Similarly, setting $k \geq 1$ such that $q_{k+2} > 0$, gives that 
\begin{equation}\label{eq:wdotplusk}
\lim_{V\to\infty} W_\bullet^{(l+k,\geq V)}/W_{\bullet}^{(l,\geq V)} < \infty 
\end{equation}
exists for all $l\geq 1$.
Combining these limits we find that
\begin{equation}\label{eq:limwdotratio}
\lim_{V\to\infty} \frac{W_{\bullet}^{(l,\geq V-k)}}{W_{\bullet}^{(l,\geq V)}} = \lim_{V\to\infty} \frac{W_{\bullet}^{(l+k,\geq V)}}{W_{\bullet}^{(l,\geq V)}}\frac{W_{\bullet}^{(l+2k,\geq V)}}{W_{\bullet}^{(l+k,\geq V)}}\frac{W_{\bullet}^{(l+2k-2,\geq V-1)}}{W_{\bullet}^{(l+2k,\geq V)}}\cdots \frac{W_{\bullet}^{(l,\geq V-k)}}{W_{\bullet}^{(l+2,\geq V-k+1)}}
\end{equation}
converges for all $l\geq 1$ and takes value in $[1,\infty)$.
It is not hard to see that if this limit is larger than one then the generating function $W_g^{(l)}$ must have radius of convergence larger than one, in contradiction with Lemma \ref{thm:rovdisk}.
Therefore we conclude that (\ref{eq:limwdotratio}) has limit one for all $l\geq 1$ and all $k\geq 1$ such that $q_{k+2}>0$. 
Since $W_\bullet^{(l,\geq V)}$ is by construction non-increasing for increasing $V$, this implies that the same limit holds for any $k\in\Z$, i.e. ,
\begin{equation}\label{eq:wdotvolconv}
\lim_{V\to\infty} \frac{W_\bullet^{(l,\geq V+v)}}{W_\bullet^{(l,\geq V)}} = 1 \quad\text{for all }l\geq 1, v\in\Z.
\end{equation}

If $\qseq$ is non-bipartite (\ref{eq:wdotmin2}) and (\ref{eq:wdotplusk}) together with (\ref{eq:wdotvolconv}) imply that for all $l,l'\geq 1$ and $v\in\Z$,
\begin{equation*}
\lim_{V\to\infty} \frac{W_\bullet^{(l',\geq V+v)}}{W_{\bullet}^{(l,\geq V)}} = \frac{F(l')}{F(l)} c_+^{l'-l}, \quad\text{with }F(l):= \lim_{V\to\infty} \frac{W_\bullet^{(l,\geq V)}}{W_{\bullet}^{(1,\geq V)}} c_+^{1-l}.
\end{equation*}
In the bipartite case we have the same identity for even $l,l'\geq 2$ and $v\in\Z$, but one should define $F(l):=  \lim_{V\to\infty} W_\bullet^{(l,\geq V)}/W_{\bullet}^{(2,\geq V)} c_+^{2-l}$.
In either case we set $F(l)=0$ for $l\leq 0$. 
By first taking the $V\to\infty$ limit of (\ref{eq:lenvolprob}) and then summing over $v\in\Z_{\geq0}$ we find
\begin{equation}\label{eq:inftransprob}
\prob_{\map_\infty}(l_{i+1} = l+k, V_{i+1} < \infty | l_i =l,V_i=v) = \frac{F(l+k)}{F(l)}\nu(k)
\end{equation}
for all $l\geq 1$ (and $l$ even in the bipartite case) and all $k\in \Z$.
The one-endedness of $\map_\infty$ implies that $\prob_{\map_\infty}(V_{i+1}=\infty | l_i=l, V_i=v)=0$ and therefore (\ref{eq:inftransprob}) leads to
\begin{equation}\label{eq:inftransprob2}
\prob_{\map_\infty}(l_{i+1} = l+k | l_i =l) = \frac{F(l+k)}{F(l)}\nu(k).
\end{equation}
Since $\map_\infty$ has no infinite faces, these probabilities must sum to one, implying that $F$ is $\nu$-harmonic on $\Z_{>0}$.
It is known, see e.g. \cite{doney_martin_1998}, Theorem 1, that for a general aperiodic random walk with step probabilities $\nu$ there exists up to rescaling at most one non-negative function that is $\nu$-harmonic on $\Z_{>0}$ and vanishes on $\Z_{\leq 0}$.
Therefore we necessarily have $F = h_r^{(1)}$, since $h_r^{(1)}$ is also $\nu$-harmonic on $\Z_{>0}$, while $F(1)=h_r^{(1)}(1)=1$ for non-bipartite $\qseq$ and $F(2)=h_1^{(1)}(2)=1$ for bipartite $\qseq$.
\end{proof}

It follows easily from the ingredients of this proof, e.g. by deducing the laws of the peeling process from the large-$V$ limit of (\ref{eq:wbulletV}), that apart from the bias in the perimeter $l_i$ the laws of the explored maps after $i$ steps in the $\qseq$-IBPM $\map_\infty$ and in the pointed $\qseq$-Boltzmann planar map $\map$ agree (provided both arise from the same peeling algorithm).
Indeed, if $\emap_0$ is a possible explored map of $\map_\infty$ after $i$ steps, i.e. $\prob(\emap(\map_\infty,F_i)=\emap_0)>0$, then 
\begin{equation}\label{eq:emaprel}
\prob(\emap(\map_\infty,F_i)=\emap_0) = \frac{h^{(1)}_r(l)}{h^{(1)}_r(l_0)} \frac{h^{(0)}_r(l_0)}{h^{(0)}_r(l)} \prob(\emap(\map,F_i)=\emap_0),
\end{equation}
where $l$ is the outer face degree of $\emap_0$ and $l_0$ is the root face degree (which we usually take to be 2 for the $\qseq$-IBPM).
In particular, this means that conditionally on the perimeter process the description of the volume process in Section \ref{sec:volume} remains valid for the infinite $\qseq$-Boltzmann planar map without change.

Theorem \ref{thm:peelingibpm} also allows us to establish the following expected limits in an indirect way.
Recall that $W^{(l,V)} = \sum_{\map}w_\qseq(\map)$ is the disk function involving a sum over planar maps $\map$ with exactly $V$ vertices and root face degree $l$.

\begin{corollary} For $l\geq 1$ and $v\in\Z$, provided $W^{(2,V-v)} \neq 0$ and $W^{(l,V)}\neq 0$ for infinitely many $V>0$, along this subsequence the disk function satisfies the limit
\begin{equation}\label{eq:fixedvlim}
\frac{W^{(l,V-v)}}{W^{(2,V)}} \xrightarrow{V\to\infty} \frac{h_r^{(1)}(l)}{h_r^{(1)}(2)} c_+^{l-2}.
\end{equation}
\end{corollary}
\begin{proof}
It is not too hard to see that provided $v$ is sufficiently large one can find a possible explored map $\emap_0$ of $\map_\infty$, in the sense as above, with outer face degree $l$ and $v$ fully explored vertices, i.e. vertices that are not incident to the outer face.
According to Theorem \ref{thm:stephenson}, $\prob(\emap(\map_\infty,F_i)=\emap_0) = \lim_{V\to\infty} \prob(\emap(\map_V,F_i)=\emap_0)$ along $V\in I_{\qseq,2}$, where $\map_V$ is a $\qseq$-Boltzmann planar map with root face degree two conditioned to have $V$ vertices.
Assuming a deterministic peeling algorithm, we can explicitly evaluate
\begin{equation*}
\prob(\emap(\map_V,F_i)=\emap_0) = w_\qseq(\emap_0)\frac{W_\bullet^{(l,V-v)}}{W_\bullet^{(2,V)}},
\end{equation*}
where $w_\qseq(\emap_0) = \prod_f q_{\deg(f)}$ with the product running over all faces except the root face and outer face.
On the other hand, the right-hand side of (\ref{eq:emaprel}) equals
\begin{equation*}
\frac{h^{(1)}_r(l)}{h^{(1)}_r(2)} \frac{h^{(0)}_r(2)}{h^{(0)}_r(l)} w_\qseq(\emap_0)\frac{W_\bullet^{(l)}}{W_\bullet^{(2)}} = \frac{h^{(1)}_r(l)}{h^{(1)}_r(2)} w_\qseq(\emap_0) c_+^{l-2}.
\end{equation*}
Combining these facts we establish (\ref{eq:fixedvlim}) for sufficiently large $v$. 
By considering suitable ratios of the established limits one may obtain the limit in the full range of $v$, but we leave the details to the reader.
\end{proof}

\section{Scaling limit of the perimeter process}\label{sec:scaleperim}

Now that we have a simple characterization of the perimeter process $(l_i)_{i\geq 0}$ of an infinite $\qseq$-Boltzmann planar map, we can try to obtain its scaling limit for general weight sequences $\qseq$.
For this we need to study the positive and negative tails of the step probabilities $\nu$, which are related to each other through the following result.

\begin{lemma}\label{thm:nukernel}
Suppose $\nu$ is critical, i.e. $h_r^{(1)}$ is $\nu$-harmonic on $\Z_{>0}$, then the negative jump probabilities can be expressed linearly in the positive probabilities as
\begin{equation}\label{eq:negfrompos}
\nu(-k) = \sum_{m=1}^{\infty} \mathcal{R}_r(k,m) \nu(m),\quad \text{for }k\geq 1,
\end{equation}
where 
\begin{align}\label{eq:negposkernel}
\mathcal{R}_r(k,m) := \sum_{p=0}^{m-1}h_r^{(1)}(m-p)\left(h_r^{(-2)}(k+p-1)+r h_r^{(-2)}(k+p-2)\right)
\end{align}
\end{lemma}
\begin{proof}
From (\ref{eq:nudef}) and (\ref{eq:usualdiskfun}) it follows that the probability generating function $\hat{\nu}(y):=\sum_{k=-\infty}^{\infty}\nu(k)y^k$, which converges on an annulus $1\leq |y|\leq R$ for some $R \in [1,\infty]$, is given formally by
\begin{equation*}
\hat{\nu}(y) = 1+\frac{1}{y} M(y c_+)\sqrt{(y-1)(y+r)},
\end{equation*}
where $c_+ = \sqrt{2 / \nu(-2)}$.
Using (\ref{eq:mfunction}) this evaluates to
\begin{equation*}
\hat{\nu}(y) = 1+ \frac{1}{y}\left(-1+\sum_{m=0}^{\infty}\sum_{l=0}^{m} \nu(m) h_r^{(0)}(m-l) y^l\right) \sqrt{(y-1)(y+r)}. 
\end{equation*}
Writing 
\begin{equation*}
\sqrt{(y-1)(y+r)} = (y+r)\sum_{p=-1}^{\infty}h_r^{(-1)}(p)y^{-p-1} = \sum_{p=-2}^{\infty}\left(h_r^{(-1)}(p+1)+rh_r^{(-1)}(p)\right)y^{-p-1}
\end{equation*}
we get for $k\geq 1$,
\begin{align*}
\nu(-k) =& [y^{-k+1}]\left(-1+\sum_{m=0}^{\infty}\sum_{l=0}^{m} \nu(m) h_r^{(0)}(m-l) y^l\right)\sum_{p=-2}^{\infty}\left(h_r^{(-1)}(p+1)+rh_r^{(-1)}(p)\right)y^{-p-1}\nonumber\\
=&-h_r^{(-1)}(k-1)-rh_r^{(-1)}(k-2) \nonumber\\
&+ \sum_{m=0}^{\infty}\sum_{l=0}^{m}h_r^{(0)}(m-l)\left(h_r^{(-1)}(l+k-1)+rh_r^{(-1)}(l+k-2)\right)\nu(m).
\end{align*}
It follows from the $\nu$-harmonicity of $h_r^{(1)}$ that 
\begin{equation*}
\nu(0) = 1-\sum_{m=1}^{\infty}h_r^{(1)}(m+1)\nu(m).
\end{equation*} 
Hence we have (\ref{eq:negfrompos}) with 
\begin{align*}
\mathcal{R}_r(k,m) =& - h_r^{(1)}(m+1) \left(h_r^{(-1)}(l+k-1)+rh_r^{(-1)}(l+k-2)\right)\nonumber\\
& + \sum_{l=0}^{m}h_r^{(0)}(m-l)\left(h_r^{(-1)}(l+k-1)+rh_r^{(-1)}(l+k-2)\right)
\end{align*}
and (\ref{eq:negposkernel}) follows from writing $h_r^{(0)}(m-l)=h_r^{(1)}(m-l+1)-h_r^{(1)}(m-l)$ and recombining the terms.
\end{proof}

In the bipartite case $r=1$, the kernel $\mathcal{R}_r(k,m)$ can further be simplified to
\begin{equation*}
\mathcal{R}_1(2k,2m) = \frac{m\, h_1^{(1)}(2k)h_1^{(1)}(2m+2)}{2k(2k-1)(m+k)} = 4^{-m-k}\frac{m(2m+1)}{(m+k)(2k-1)}\binom{2k}{k}\binom{2m}{m}.
\end{equation*}

From (\ref{eq:negposkernel}) we can easily determine the large-$k$ behaviour of $\mathcal{R}_r(k,m)$ for fixed $m$,
\begin{equation*}
\lim_{k\to\infty}\frac{\mathcal{R}_r(k,m)}{h_r^{(-2)}(k)} = (1+r)\sum_{p=0}^{m-1}h_r^{(1)}(m-p) = (1+r)h_r^{(2)}(m+1),
\end{equation*}
where in the bipartite case $k$ and $m$ are restricted to be even.
Recall from (\ref{eq:lenconst}) the definition $\lenconst_\nu = \sum_{k=1}^\infty \nu(k) h_r^{(2)}(k+1)$, which is not necessarily finite.
If it is finite, for which regular criticality of $\nu$ is a sufficient condition, then Lemma \ref{thm:nukernel} together with (\ref{eq:hfunasymp}) implies that for $k\to\infty$ (and $r<1$)
\begin{equation}\label{eq:negnuasymp}
\lim_{k\to\infty}\nu(-k)\, k^{5/2} = \lenconst_\nu (1+r) \lim_{k\to\infty}h_r^{(-2)}(k)\,k^{5/2} = \frac{3\lenconst_\nu\sqrt{1+r}}{4\sqrt{\pi}}.
\end{equation}
If $r=1$, $\nu(-k)$ vanishes for odd $k$ but the limit still holds when $\nu(-k)$ is replaced by $(\nu(-k)+\nu(-k-1))/2$.

\begin{lemma}\label{thm:charfun}
If $h_r^{(1)}$ is $\nu$-harmonic on $\Z_{>0}$ and $\lenconst_\nu  < \infty$, then the characteristic function $\varphi_{\nu}(\theta)$ close to $\theta=0$ satisfies
\begin{equation}\label{eq:charfunexp}
\varphi_{\nu}(\theta):= \sum_{k=-\infty}^{\infty} \nu(k) e^{ik \theta} =  1 - \sqrt{\frac{1+r}{2}} \lenconst_\nu |\theta|^{1/2}(|\theta|-i\theta) + \mathcal{O}(|\theta|^{5/2}),
\end{equation}
which implies in particular that $\nu$ is centered.
\end{lemma}
\begin{proof}
From $\lenconst_\nu  < \infty$ and (\ref{eq:negnuasymp}) it follows that $\nu$ has finite first moments and therefore the characteristic function $\varphi_{\nu}(\theta)$ is continuously differentiable.
Using (\ref{eq:nudef}) and (\ref{eq:usualdiskfun}) we can express $\varphi_{\nu}(\theta)$ as
\begin{equation}\label{eq:charfunexpr}
\varphi_{\nu}(\theta) = 1 + e^{-i\theta}M(c_+ e^{i\theta})\sqrt{(e^{i\theta}-1)(e^{i\theta}+r)}.
\end{equation}
From (\ref{eq:mfunction}) it follows that
\begin{equation*}
M(c_+e^{i\theta}) = -1 + \sum_{m=0}^\infty \nu(m) \sum_{l=0}^m h_r^{(0)}(m-l) e^{i l\theta} = \sum_{m=0}^\infty \nu(m) \sum_{l=0}^m h_r^{(0)}(m-l) \left(e^{i l\theta}-1\right).
\end{equation*}
One may check that for any $\theta\in\R$ and $l \geq 0$ 
\begin{equation*}
\left|(e^{il\theta}-1) \sqrt{(e^{i\theta}-1)(e^{i\theta}+r)}\right| \leq l\sqrt{1+r}\, |\theta|^{3/2}.
\end{equation*}
Therefore 
\begin{equation*}
|\varphi_{\nu}(\theta)-1| \leq \sqrt{1+r}|\theta|^{3/2}\sum_{m=0}^{\infty}\nu(m) \sum_{l=0}^{m}l\,h_r^{(0)}(m-l) = \sqrt{1+r}|\theta|^{3/2} \lenconst_\nu
\end{equation*}
and we may perform the small-$\theta$ expansion termwise.
Equation (\ref{eq:charfunexp}) then follows from 
\begin{equation*}
e^{-i\theta}(e^{il\theta}-1) \sqrt{(e^{i\theta}-1)(e^{i\theta}+r)} = -\sqrt{\frac{1+r}{2}}|\theta|^{1/2}(|\theta|-i\theta) + \mathcal{O}(|\theta|^{5/2}).
\end{equation*}
In particular $\varphi_\nu'(0) = 0$ and therefore $\sum_{k=-\infty}^\infty k\nu(k) = 0$.
\end{proof}
Notice that the last statement of Lemma \ref{thm:charfun} that $\nu$ is centered can also be inferred from Proposition \ref{thm:drift}, since an oscillating random walk $(X_i)_i$ with finite first moments necessarily has vanishing drift. 

Let $S_{3/2}$ be the $3/2$-stable Levy process with no positive jumps, normalized such that for all $\theta\in\R$
\begin{equation}
\expec \exp(i \theta S_{3/2}(t)) = \exp\left[-t|\theta|^{1/2}(|\theta|-i\theta)/\sqrt{2}\right],
\end{equation}
and $S_{3/2}^+$ the corresponding process conditioned to be positive (see e.g. \cite{chaumont_conditionings_1996}).
The following result is the analogue of Proposition 5 in \cite{curien_scaling_2014}.

\begin{proposition}\label{thm:lengthscaling}
Let $(X_k)_{k\geq 0}$ be a random walk with step probabilities $\nu(k)$ and $r\in(-1,1]$, such that $h_r^{(1)}$ is $\nu$-harmonic on $\Z_{>0}$ and let $(l_k)_{k\geq 0}$ be the corresponding Doob transform.
If $\lenconst_\nu < \infty$ then we have the following convergence in distribution in the sense of Skorokhod of $(X_k)_{k\geq 0}$ and $(l_k)_{k\geq 0}$ to $S_{3/2}$ and $S_{3/2}^+$,
\begin{equation}
\left( \frac{X_{\lfloor n t\rfloor}}{\left(\sqrt{1+r}\lenconst_\nu n\right)^{\frac{2}{3}}} \right)_{t\geq 0} \xrightarrow[n\to\infty]{(\rmd)} S_{3/2} (t),\quad\quad
\left( \frac{l_{\lfloor n t\rfloor}}{\left(\sqrt{1+r}\lenconst_\nu n\right)^{\frac{2}{3}}} \right)_{t\geq 0} \xrightarrow[n\to\infty]{(\rmd)} S^+_{3/2} (t).
\end{equation}
\end{proposition}
\begin{proof}
The characteristic function for the distribution of $n^{-2/3}X_{\lfloor n t\rfloor}$ is given by
\begin{equation*}
\theta \to \left(\varphi_\nu(n^{-2/3}\theta)\right)^{\lfloor n t\rfloor}
\end{equation*}
with $\varphi_\nu$ as in Lemma \ref{thm:charfun}.
It converges pointwise as $n\to\infty$ to
\begin{equation*}
\left(\varphi_\nu(n^{-2/3}\theta)\right)^{\lfloor n t\rfloor} \xrightarrow[n\to\infty]{} \exp\left[-t\sqrt{\frac{1+r}{2}} \lenconst_\nu |\theta|^{1/2}(|\theta|-i\theta)\right],
\end{equation*}
which implies the stated convergence in distribution of the random walk $(X_k)_{k\geq 0}$.
Since $(l_k)_{k\geq 0}$ is obtained from $(X_k)_{k\geq 0}$ by conditioning on staying positive its convergence to $S^+_{3/2}$ follows from the invariance principle in \cite{caravenna_invariance_2008} (see also \cite{curien_scaling_2014}, Proposition 5, for the analogous statement). 
\end{proof}

\subsection{Intermezzo: heavy-tailed case}\label{sec:heavytailed}
In the rest of the paper we will only consider cases in which $\lenconst_\nu < \infty$, but for the sake of completeness let us comment on what happens when $\lenconst_\nu = \infty$.
The latter implies that $\sum_{k=1}^\infty k^{3/2} \nu(k)=\infty$, while $\sum_{k=1}^\infty k^{1/2} \nu(k)<\infty$ by (\ref{eq:admissiblecondition}).
Therefore we may define
\begin{equation*}
\alpha = \inf\left\{s\in\R:\sum_{k=1}^{\infty} k^{s} \nu(k) = \infty \right\}
\end{equation*}
which necessarily takes values in $[1/2,3/2]$.
For simplicity let us assume $r<1$ and
\begin{equation*}
\lim_{k\to\infty} k^{\alpha}\sum_{m=k}^{\infty} \nu(m)= p_+
\end{equation*}
with $p_+$ finite and positive. 
Then
\begin{align}
\lim_{k\to\infty} k^{\alpha+1}\nu(-k) &= \lim_{k\to\infty} \sum_{m=1}^{\infty} k^{\alpha+1} \left(\mathcal{R}_r(k,m)-\mathcal{R}_r(k,m-1)\right)\sum_{l=m}^{\infty}\nu(l) \nonumber\\
&= \lim_{k\to\infty} \int_0^\infty \rmd x k^2 \left(\mathcal{R}_r(k,\lfloor x k \rfloor)-\mathcal{R}_r(k,\lfloor x k \rfloor-1)\right) k^{\alpha}\sum_{l=\lfloor x k \rfloor}^{\infty}\nu(l).\label{eq:negnuasymp2}
\end{align}
One may check using the asymptotics (\ref{eq:hfunasymp}) that for $r<1$ we have 
\begin{equation*}
\lim_{k\to\infty} k^2\left(\mathcal{R}_r(k,\lfloor x k \rfloor)- \mathcal{R}_r(k,\lfloor x k \rfloor-1)\right) = \frac{\sqrt{x}}{2\pi} \frac{x+3}{(x+1)^2},
\end{equation*}
which suggests that (\ref{eq:negnuasymp2}) should equal
\begin{equation*}
p_+\int_0^\infty \rmd x \frac{\sqrt{x}}{2\pi}\frac{x+3}{(x+1)^2} x^{-\alpha} = \frac{\alpha}{\cos(\pi(\alpha-1))} p_+.
\end{equation*}
Hence, both tails of $\nu$ fall off with the same power $-\alpha-1$ and the positive tail is lighter by a factor of $\cos(\pi(\alpha-1))$ compared to the negative tail.\footnote{In a slightly different form this relation was observed in \cite{borot_recursive_2012,borot_more_2012} in the setting of planar maps coupled to an $O(n)$ loop model with $0<n<1$. The \emph{gasket} of a random map with loops, which is obtained by removing the interiors of all loops, turns out to be distributed as a $\qseq$-Boltzmann planar map for a particular weight sequence satisfying $q_k \sim k^{-\alpha-1}c_+^{-k}$, where $1/2<\alpha<3/2$ is one of the two solutions to $n=2\cos(\pi(\alpha-1))$. Notice that the ratio of the tails of $\nu$ is precisely $n/2$.}
We therefore expect, but do not prove, that $(n^{-1/\alpha}X_{\lfloor n t\rfloor})_{t\geq 0}$ converges in distribution as $n\to\infty$ (up to normalization) to an $\alpha$-stable Levy process $S_{\alpha,\beta}$ with skewness parameter 
\begin{equation}\label{eq:skewness}
\beta = -\cot^2(\pi\alpha/2),
\end{equation}
which is characterized by\footnote{Notice that the normalization of the process $S_{3/2}$ in Proposition \ref{thm:lengthscaling} differs slightly from $S_{\alpha,\beta}$ with $\alpha=3/2$ and $\beta=-1$: one has $S_{3/2}(\sqrt{2}t)=S_{3/2,-1}(t)$.} 
\begin{align}
\expec \exp(i \theta S_{\alpha,\beta}(t)) &= \exp\left[-t|\theta|^{\alpha}(1-i \beta \tan(\pi\alpha/2) \mathrm{sgn}(\theta))\right] \nonumber\\
&= \exp\left[-t|\theta|^{\alpha}(1-i \cot(\pi\alpha/2) \mathrm{sgn}(\theta))\right].
\end{align}
Another way to arrive at (\ref{eq:skewness}) is by combining Zolotarev's formula \cite{zolotarev_one-dimensional_1986} for the \emph{positivity parameter}
\begin{equation*}
\rho := \prob( S_{\alpha,\beta}(1) \geq 0 ) = \frac{1}{2} +\frac{1}{\pi\alpha}\arctan\left(\beta\tan\frac{\pi\alpha}{2}\right)
\end{equation*}
with the fact (see, e.g., \cite{chaumont_conditionings_1996,caravenna_invariance_2008}) that the renewal function $x\to x^{\alpha(1-\rho)}$ for $S_{\alpha,\beta}$ should be the scaling limit of the discrete renewal function $h_r^{(1)}$, implying that $\alpha(1-\rho) = 1/2$.

Notice that the value $\alpha=1$ is special since it yields a symmetric limiting process, namely the Cauchy process $S_{1,0}$.
One may wonder whether a random walk with critical step probabilities $\nu$ exists which is already symmetric at the discrete level.
It turns out that for each $r\in(-1,1]$ there exists a one parameter family of such random walks.
To see this recall the expression (\ref{eq:charfunexpr}) for the characteristic function $\varphi_\nu(\theta)$.
Symmetry of $\nu$ is equivalent to $\varphi_\nu(\theta)$ being real for all $\theta$, i.e.
\begin{equation}\label{eq:realcond}
M(c_+ e^{i\theta})e^{-i\theta}\sqrt{(e^{i\theta}-1)(e^{i\theta}+r)} \in \R.
\end{equation}
Notice that $e^{-i\theta}\sqrt{(e^{i\theta}-1)(e^{i\theta}+r)}$ only contains non-positive powers of $e^{i\theta}$, while $M(c_+ e^{i\theta})$ contains only non-negative powers of $e^{i\theta}$.
Therefore the only way (\ref{eq:realcond}) can be real is when they are complex conjugates up to a real constant.
Hence
\begin{equation*}
\varphi_\nu(\theta) = 1 - a \left|\sqrt{(e^{i\theta}-1)(e^{i\theta}+r)}\right|^2 = 1-2a \sqrt{1+r^2+2r \cos(\theta)}\left|\sin(\theta/2)\right|.
\end{equation*}
This corresponds to a (critical) step probability distribution $\nu$ if and only if
\begin{equation*}
\int_0^{2\pi} \varphi_\nu(\theta) \in [0,1),
\end{equation*}
which amounts to
\begin{equation*}
0 < a \leq \begin{cases}
\pi/4 & \text{for }r=0,1\\
\frac{\pi}{2(r+1)+\frac{(r-1)^2}{\sqrt{r}}\mathrm{arctanh}\left(\frac{2\sqrt{r}}{r+1}\right)} & \text{for }0<r< 1 \\
\frac{\pi}{2(r+1)+\frac{(r-1)^2}{\sqrt{-r}}\mathrm{arctan}\left(\frac{2\sqrt{-r}}{r+1}\right)} & \text{for }-1<r<0
\end{cases}.
\end{equation*}
For example, in the case $r=1$ and $a=\pi/4$ we get
\begin{equation*}
\nu(k)=\begin{cases}\frac{1}{k^2-1} & \text{for }k\text{ even and }k\neq 0   \\
	0 & \text{otherwise}
	\end{cases},
\end{equation*}
which is the random walk associated to the weight sequence with $q_{2k} = 6^{1-k}/((2k-2)^2-1)$, $k>1$, and the other weights zero. 

\section{Scaling limit of the volume process}\label{sec:scalevol}
Let $\qseq$ be regular critical and $\map_l\in\maps^{(l)}$, $l\geq 1$, a $\qseq$-Boltzmann planar map with root face degree $l$.
The expected number of vertices in $\map_l$ is easily determined to be
\begin{equation*}
\expec|\mathcal{V}(\map_l)| = \frac{W_\bullet^{(l)}}{W^{(l)}} = \frac{h_r^{(0)}(l)\nu(-2)}{\nu(-l-2)}.
\end{equation*}
Therefore, using (\ref{eq:negnuasymp}),
\begin{equation}\label{eq:volexpeclim}
\mathcal{B}_\nu:=\lim_{l\to\infty} l^{-2} \expec|\mathcal{V}(\map_l)| = \frac{\nu(-2)}{(1+r)\lenconst_\nu} \lim_{l\to\infty} \frac{h_r^{(0)}(l)}{l^2h_r^{(-2)}(l)} = \frac{4\nu(-2)}{3(1+r)\lenconst_\nu}.
\end{equation}
With some extra work we can determine the convergence in distribution of $l^{-2}|\mathcal{V}(\map_l)|$ as follows.

\begin{proposition}\label{thm:volumedist}
Let $\xi$ be a positive random variable with density $e^{-1/(2\xi)} \xi^{-5/2} /\sqrt{2\pi}$. 
Then we have the weak convergence
\begin{equation}\label{eq:volumeconv}
\frac{l^{-2}|\mathcal{V}(\map_l)|}{\mathcal{B}_{\nu}} \xrightarrow[l\to\infty]{(d)} \xi.
\end{equation}
\end{proposition}
\begin{proof}
We will prove that the Laplace transform of the left-hand side of (\ref{eq:volumeconv}) converges pointwise to the Laplace transform of $\xi$, i.e.,
\begin{equation}\label{eq:volumelaplacelimit}
\lim_{l\to\infty}
\expec \exp\left(-\lambda l^{-2}|\mathcal{V}(\map_l)|/\mathcal{B}_\nu\right) = (1+\sqrt{2\lambda})e^{-\sqrt{2\lambda}}
\end{equation}
for all $\lambda > 0$.
In order to evaluate the expectation value on the left-hand side let us write $g=\exp(-\lambda l^{-2}/\mathcal{B}_\nu)<1$
and consider the disk function $W^{(l)}_g$ in (\ref{eq:diskfunctionvertexweight}).
Recall that we may identify $W^{(l)}_g(\qseq) = g^{1+l/2} W^{(l)}(\qseq_g)$ where $\qseq_g$ is the weight sequence determined by $(q_g)_{k} := g^{(k-2)/2}q_k$ and let $c_{\pm}(g)$ (and $r(g) = -c_-(g)/c_+(g)$) be the associated constants from Proposition \ref{thm:pointeddisk}.
Then one finds
\begin{equation*}
\frac{\expec|\mathcal{V}(\map_l)| g^{|\mathcal{V}(\map_l)|}}{\expec |\mathcal{V}(\map_l)|} = \frac{W_{\bullet,g}^{(l)}}{W_\bullet^{(l)}} = \left(\frac{c_+(g)}{c_+(1)}\right)^l \frac{h^{(0)}_{r(g)}(l)}{h^{(0)}_{r(1)}(l)}. 
\end{equation*}
From Lemma \ref{thm:volumeexpansion} it follows that $r(g)$ is continuous as $g\to 1$ and that
\begin{equation*}
\lim_{g\uparrow 1} \frac{1-c_+(g)/c_+(1)}{\sqrt{1-g}} = \sqrt{2\mathcal{B}_{\nu}}.
\end{equation*}
Setting $g=\exp(-\lambda l^{-2}/\mathcal{B}_\nu)$ and using the asymptotics (\ref{eq:hfunasymp}) we find 
\begin{equation*}
\lim_{l\to\infty} h^{(0)}_{r(g)}(l)/h^{(0)}_{r(1)}(l) = 1
\end{equation*}
and therefore
\begin{equation*}
\lim_{l\to\infty} \frac{\expec|\mathcal{V}(\map_l)| \exp(-\lambda l^{-2}|\mathcal{V}(\map_l)|/\mathcal{B}_\nu)}{\expec|\mathcal{V}(\map_l)|} = \lim_{l\to\infty} \left(1-\sqrt{2\lambda}/l\right)^l = e^{-\sqrt{2\lambda}}.
\end{equation*}
Since the left-hand side is bounded by one uniformly in $\lambda$ we may integrate both sides with respect to $\lambda$, which leads to (\ref{eq:volumelaplacelimit}).
\end{proof}

Given the process $S_{3/2}^+$ let us introduce an associated process $Z$ as in \cite{curien_scaling_2014,curien_hull_2014} by setting
\begin{equation*}
Z(t) = \sum_{t_i \leq t} \xi_i (\Delta S_{3/2}^+(t_i))^2,
\end{equation*}
where $(t_i)_{i\geq 0}$ is a measurable enumeration of the jumps of $S_{3/2}^+(t)$, $\Delta S_{3/2}^+(t_i)$ is the size of the jump at time $t_i$, and $(\xi_i)_{i\geq 0}$ are independent random variables with the law of $\xi$ in Proposition \ref{thm:volumedist}.

\begin{theorem}[Analogue of Curien--le Gall \cite{curien_scaling_2014}, Theorem 1]
Let $\qseq$ be a regular critical weight sequence. 
Then one has the following joint convergence in distribution in the sense of Skorokhod of the boundary length $(l_k)_{k\geq0}$ and the volume $(V_k)_{k\geq 0}$ of the peeling process of the infinite $\qseq$-Boltzmann planar map:
\begin{equation}
\left(\frac{l_{\lfloor n t\rfloor}}{\left(\sqrt{1+r}\lenconst_\nu n\right)^{2/3}}, \frac{V_{\lfloor n t\rfloor}}{\frac{8}{3c_+^2}\left(\frac{\lenconst_\nu}{1+r}\right)^{1/3} n^{4/3}}\right)_{t\geq 0} \xrightarrow[n\to\infty]{(\rmd)} \left(S^+_{3/2}(t),Z(t)\right)_{t\geq 0},
\end{equation}
\end{theorem}
\begin{proof}
As in \cite{curien_scaling_2014} we may write
\begin{equation*}
V_k = \sum_{i=1}^{k} \mathbf{1}_{\{l_i < l_{i-1}-1\}} U_i,
\end{equation*}
where $U_i$ are independent random variables distributed like $\mathcal{V}(l_{i-1}-l_i-2)$.
The proof of \cite{curien_scaling_2014} goes through without modification, since it only depends on (the analogues of)  Propositions \ref{thm:lengthscaling} and \ref{thm:volumedist} and on the following two bounds.
It follows from (\ref{eq:volexpeclim}) and (\ref{eq:negnuasymp}) that there exist constants $C,C'>0$ such that $\expec\mathcal{V}(l) \leq C l^2$ for all $l\geq 0$ and $h_r^{(1)}(l-k)\nu(-k)/h_r^{(1)}(l) \leq \nu(-k) \leq C' k^{-5/2}$ for all $l,k\geq 1$.
\end{proof}

\section{Examples}\label{sec:examples}

The one-to-one relation between critical weight sequences and random walks with jump probabilities $\nu$ such that $h_r^{(1)}$ is $\nu$-harmonic on $\Z_{>0}$ for some $r$ provides a simple algorithm to determine the scaling constants for specific infinite $\qseq$-Boltzmann planar maps.
First one imposes desired conditions on the $\nu(k)$ for $k\geq -1$. 
Then the $\nu$-harmonicity of $h_r^{(1)}$ on $\{1,2,3\}$ gives three equations, of which two can be used to fix $r$ and $c_+$, while the remaining corresponds to the criticality condition.
The $\nu$-harmonicity of $h_r^{(1)}$ on $\Z_{\geq 4}$ then allows one to solve for all $\nu(k)$ with $k\leq -3$ (perhaps by using Lemma \ref{thm:nukernel}).
Notice that the latter encode exactly the (unpointed) disk functions $W^{(l)}$.

\subsection{$2p$-angulations}
Let $p\geq 2$ and set $\nu(k) = 0$ for all $-1 \leq k \neq 2p-2$, while $\nu(2p-2)>0$.
Since $2p$-angulations are bipartite, we necessarily have $r=1$. 
It follows from the $\nu$-harmonicity of $h_1^{(1)}$ on $\{1,3\}$ that
\begin{equation*}
\nu(2p-2) = \frac{1}{h_1^{(1)}(2p-1)} = \frac{2^{2p-1}}{p \binom{2p}{p}},\quad  \nu(-2)=h_1^{(1)}(3)-h_1^{(1)}(2p+1)\nu(2p-2) = \frac{p-1}{2p}
\end{equation*}
and therefore $c_+$ and the critical weight $q_{2p}$ are given by
\begin{equation*}
c_+ = \sqrt{\frac{2}{\nu(-2)}} = \sqrt{\frac{4p}{p-1}}\quad \text{and}\quad q_{2p} = c_+^{-2p+2}\nu(2p-2)=2\left(1-\frac{1}{p}\right)^{p-1}\frac{1}{p\binom{2p}{p}}.
\end{equation*}
After some algebra one finds for $l>0$,
\begin{equation*}
\nu(-2l) = \frac{p-1}{2l(2l-1)(p+l-1)}h_1^{(1)}(2l) = 4^{-l}\frac{p-1}{(p+l-1)(2l-1)}\binom{2l}{l}
\end{equation*}
and
\begin{equation*}
\lenconst_\nu = \sum_{k=1}^{\infty}h_1^{(2)}(k+1)\nu(k) = \frac{4}{3}(p-1). 
\end{equation*}
\subsection{$(2p+1)$-angulations} 
This time we set $\nu(k)=0$ for $-1 \leq k \neq 2p-1$, which leads to
\begin{equation*}
\nu(2p-1) = \frac{1}{h_r^{(1)}(2p)} \quad \text{and} \quad \nu(-2) = h_r^{(1)}(3) - h_r^{(1)}(2p+2)\nu(2p-1),
\end{equation*}
while $r$ is the unique root in $(-1,1)$ of the $2p$'th order polynomial
\begin{equation*}
h_r^{(1)}(2p+1) - \frac{1}{2}(3-r)h_r^{(1)}(2p)=0.
\end{equation*}
Only the case $p=1$ admits a simple closed-form solution, leading to the well-known critical values for (type-I) triangulations,
\begin{equation*}
r =2\sqrt{3}-3, \quad
q_3 = \frac{1}{\sqrt{12\sqrt{3}}} \quad\text{and}\quad c_+ = \sqrt{6+4\sqrt{3}}.
\end{equation*}
The corresponding scaling constant for the uniform infinite planar triangulation (UIPT) then becomes
\begin{equation*}
\lenconst_\nu = \frac{1}{2}\left(1+\frac{1}{\sqrt{3}}\right).
\end{equation*}
\subsection{Uniform planar maps with controlled number of edges and vertices}
As a final example let us consider the random planar map $\map\in\maps^{(l)}$ that is sampled according to the critical Boltzmann weight $\map\to a^{|\vertices(\map)|}b^{|\edges(\map)|}$ for suitable $a,b>0$.
Using Euler's formula it is easy to see that these are precisely the $\qseq$-Boltzmann planar maps for which the critical weight sequences $\qseq$ are geometric sequences.
By Corollary \ref{thm:criticalwalk} these are in bijection with step probabilities $\nu$ for which $h_r^{(1)}$ is $\nu$-harmonic on $\Z_{>0}$ and for which $(\nu(k))_{k\geq -1}$ is geometric. 
It is straightforward to construct all the $\nu$ satisfying the latter conditions.
Indeed, we may write $\nu(k) = \alpha \sigma^k$ for $k\geq -1$ with $\alpha>0$ and $0<\sigma<1$, and we consider the $\nu$-harmonicity of $h_r^{(1)}$ on $\{1,2,3\}$.
The first condition gives
\begin{equation*}
1= \sum_{k=0}^{\infty} h_r^{(1)}(k+1)\nu(k) = \frac{\alpha}{(1-\sigma)^{3/2}\sqrt{1+r\sigma}},
\end{equation*}
fixing $\alpha$ in terms of $r$ and $\sigma$, while the second is equivalent to $(3-r)\sigma = 2$.
Since we need $r\in(0,1]$ we should restrict to $\sigma \in (1/2,1)$ and take $\alpha = (1-\sigma)^{3/2}\sqrt{3\sigma-1}$.
Finally, the third conditions reads
\begin{equation*}
\frac{3}{8}(5-2r+r^2) = \sum_{k=-1}^{\infty}h_r^{(1)}(k+3)\nu(k) = \frac{1-\alpha}{\sigma^2} +\frac{2}{c_+^2},
\end{equation*}
which allows us to express $c_+$ in terms of $\sigma$.
It is convenient here to switch parameters from $\sigma$ to $\mathcal{H}$ defined by 
\begin{equation*}
\mathcal{H} := \sum_{k=0}^{\infty} (k+1)\nu(k) = \sqrt{\frac{3\sigma-1}{1-\sigma}},\quad\quad \sigma = \frac{\hopconst^2+1}{\hopconst^2+3},
\end{equation*} 
such that $\mathcal{H}\in(1,\infty)$.
Some simple algebra allows us to express the various constants as
\begin{equation*}
r=\frac{\hopconst^2-3}{\hopconst^2+1},\quad c_+ = \frac{2(\hopconst^2+1)}{(\hopconst-1)^{3/2}\sqrt{\hopconst+3}}, \quad \lenconst_\nu = \frac{1}{2}(\hopconst^2+1).
\end{equation*}
In particular we may parameterize all critical geometric weight sequences by
\begin{equation}\label{eq:geomq}
q_k = \frac{16 \mathcal{H}}{(\mathcal{H}+3)(\mathcal{H}-1)^3}\left(\frac{(\mathcal{H}-1)^{3/2}\sqrt{\mathcal{H}+3}}{2(\mathcal{H}^2+3)}\right)^k,
\end{equation}
which may be compared with \cite{fusy_three-point_2014}, Section 5.\footnote{The parameter $r$ introduced in \cite{fusy_three-point_2014}, Section 5.1, is related to $\hopconst$ by $r=3/\hopconst$.}

The Boltzmann planar maps associated to geometric weight sequences are special in that they are the only Boltzmann planar maps whose dual maps are also distributed as Boltzmann planar maps.
Indeed, the weight $a^{|\vertices(\map)|}b^{|\edges(\map)|}$ is proportional to $(a')^{|\faces(\map)|}(b')^{|\edges(\map)|}$ for suitably chosen $a'$ and $b'$.
This duality corresponds exactly to replacing $\frac{\mathcal{H}-1}{2} \to \frac{2}{\mathcal{H}-1}$ in the weight sequence (\ref{eq:geomq}).
In particular, the self-dual case $\mathcal{H}=3$ yields the weight sequence of uniform planar maps, for the which the Boltzmann weight depends only on the number of edges. 

\appendix
\section{Proofs using results and notation of Miermont}\label{sec:app}

The starting point for the proofs to follow is the following result appearing in \cite{miermont_invariance_2006}.

\begin{proposition}[Miermont \cite{miermont_invariance_2006}, Proposition 1]\label{thm:miermont}
Define
\begin{align}
f^\bullet (x,y) &:= \sum_{k=0}^{\infty}\sum_{k'=0}^{\infty} x^k y^{k'}\binom{2k+k'+1}{k+1}\binom{k+k'}{k} q_{2+2k+k'} \label{eq:fdotdef}\\
f^\diamond(x,y) &:= \sum_{k=0}^{\infty}\sum_{k'=0}^{\infty} x^k y^{k'}\binom{2k+k'}{k}\binom{k+k'}{k} q_{1+2k+k'}\label{eq:fdiamonddef}.
\end{align}
The non-bipartite sequence $\qseq$ is admissible if and only if there exist $\zp,\zd>0$ such that
\begin{equation}\label{eq:zpzdeq}
f^{\bullet}(\zp,\zd) = 1-\frac{1}{\zp},\quad f^{\diamond}(\zp,\zd) = \zd
\end{equation}
and the matrix
\begin{equation}\label{eq:mmatrix}
\mathfrak{M}_{\qseq}(\zp,\zd) := \begin{pmatrix}
0 & 0 & \zp-1 \\
\frac{\zp}{\zd} \partial_x f^\diamond(\zp,\zd) & \partial_yf^\diamond(\zp,\zd) & 0 \\
\frac{(\zp)^2}{\zp-1}\partial_x f^\bullet(\zp,\zd) & \frac{\zp\zd}{\zp-1}\partial_yf^\bullet(\zp,\zd) & 0
\end{pmatrix}
\end{equation}
has spectral radius $\rho_\qseq \leq 1$.
Moreover, if $\qseq$ is admissible the solution $\zp,\zd$ is unique and the partition function (\ref{eq:partitionfunction}) is given by $Z_{\bullet}(\qseq) := W_\bullet^{(2)}(\qseq)-1 = 2\zp+(\zd)^2-1$.
\end{proposition}

The case of bipartite weight sequences $\qseq$ is not quite covered by this result, but was solved earlier in \cite{marckert_invariance_2007}, Proposition 1.
There it was shown that a bipartite weight sequence is admissible if and only if $f^\bullet(\zp,0)=1-1/\zp$ has a solution such that $(\zp)^2 \partial_xf^\bullet(\zp,0) \leq 1$.
Therefore Proposition \ref{thm:miermont} can be made to cover arbitrary weight sequences by allowing $\zd\geq0$ and replacing the condition on $\mathfrak{M}_{\qseq}(\zp,\zd)$ by $(\zp)^2 \partial_xf^\bullet(\zp,0) \leq 1$ whenever $\zd=0$.

\subsection{Proof of Proposition \ref{thm:pointeddisk}}\label{sec:proofpointeddisk}
\begin{proof}
Let $\qseq$ be an admissible weight sequence, i.e., $Z_{\bullet}(\qseq) := W_\bullet^{(2)}(\qseq)-1 < \infty$.
For any $(\map,v)\in\maps_{\bullet}$ one may label the vertices of $\map$ by their graph distance to the marked vertex $v$.
The set of pointed rooted maps naturally partitions into $\maps_\bullet = \maps_\bullet^+ \bigcup\maps_\bullet^-\bigcup\maps_\bullet^0$ based on whether the labels increase, decrease or remain constant along the root edge of the map.
The partition function decomposes accordingly as 
\begin{equation*}
Z_{\bullet} = Z_{\bullet}^+ + Z_{\bullet}^- + Z_{\bullet}^0 = 2Z_{\bullet}^++Z_{\bullet}^0,
\end{equation*}
where we have defined
\begin{equation*}
Z_\bullet^{\epsilon} := \sum_{\map\in\maps_\bullet^\epsilon}\prod_{f\in\faces(\map)} q_{\deg(f)}
\end{equation*}
for $\epsilon\in\{-,0,+\}$.
Clearly $Z_{\bullet}<\infty$ then implies $Z_{\bullet}^{\epsilon}<\infty$.

According to the Bouttier-Di~Francesco-Guitter bijection \cite{bouttier_planar_2004} $\zp := Z^+_\bullet+1$ is the generating function for labeled mobiles rooted at a labeled vertex, while $z^{\diamond} := \sqrt{Z_\bullet^0}$ is the generating function for so-called half-mobiles.
We will see that it is straightforward to write down an expression for the pointed disk function $W_\bullet^{(l)}$ using these generating functions.
Similar expression can be found in \cite{bouttier_planar_2004}.

\begin{figure}[t]
\begin{center}
\includegraphics[width=\linewidth]{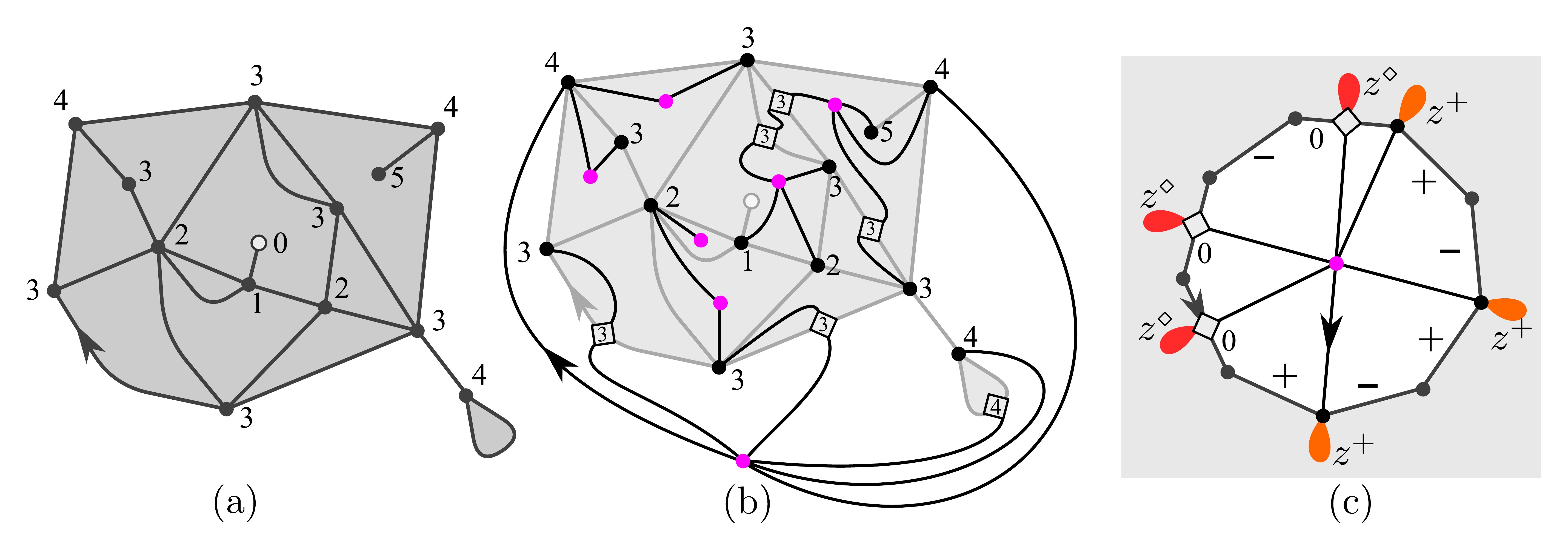}
\end{center}
\caption{The Bouttier-Di~Francesco-Guitter bijection gives a map between pointed rooted planar maps (figure a) and labeled mobiles (figure b), which are rooted planar trees with three types of labeled or unlabeled vertices. The numbers $\zp,\zd$ give the total weight of all possible submobiles appearing in a branch ending at either a black vertex or a diamond shaped vertex (figure c). The total weight of all mobiles sharing the same ``distance change sequence'' $(x_i)_{i=1}^l = (0,1,-1,1,-1,1,0,-1,0)$ is thus $(\zp)^3(\zd)^3$. } %
\label{fig:bdfg}
\end{figure}

Given $(\map,v)\in\maps_\bullet^{(l)}$, let $(e_i)_{i=1}^{l}$ be the sequence of oriented edges around the root face in counterclockwise order (or clockwise when the root face is taken to be the outer face as in figure \ref{fig:bdfg}a) starting with the root edge and define the sequence $(x_i)_{i=1}^{l}\in\{-1,0,1\}^l$ by setting $x_i$ equal to the change in the distance to $v$ along the edge $e_i$.
The labeled mobile corresponding to $(\map,v)$, rooted at the vertex corresponding to the root face of $\map$, decomposes into $k$ mobiles that are rooted at a labeled vertex and $k'$ half-mobiles, where $k$ and $k'$ are respectively the number of $1$'s and $0$'s in the sequence $(x_i)_{i=1}^{l}$ (see figure \ref{fig:bdfg}c).
Since $\sum_{i=1}^lx_i=0$, we necessarily have $k'+2k = l$.
It follows from the Bouttier-Di~Francesco-Guitter bijection that any such sequence $(x_i)_{i=1}^{l}$ satisfying the latter condition can occur.
Hence we may identify
\begin{equation*}
W_{\bullet}^{(l)} = \sum_{k=0}^{\lfloor l/2\rfloor} \frac{l!}{(k!)^2(l-2k)!} (\zp)^k (z^{\diamond})^{l-2k},
\end{equation*}
which is necessarily finite.
Its generating function is
\begin{align*}
W_{\bullet}(z) &= \sum_{l=0}^{\infty}W_{\bullet}^{(l)} z^{-l-1} = \sum_{k=0}^{\infty}\sum_{l=2k}^{\infty}\frac{1}{z}\binom{2k}{k}\binom{l}{2k}\left(\frac{\zp}{(z^{\diamond})^2}\right)^k\left(\frac{z^{\diamond}}{z}\right)^l\nonumber\\
&= \sum_{k=0}^{\infty} \binom{2k}{k} \frac{1}{z-z^{\diamond}} \left(\frac{\zp}{(z-z^{\diamond})^2}\right)^k = \frac{1}{\sqrt{(z-z^{\diamond})^2-4\zp}},
\end{align*}
where in the third equality we used the identity $\sum_{l=2k}^\infty \binom{l}{2k} y^{l-2k} = [x^{2k}](1-y-x)^{-1} =  (1-y)^{-2k-1}$ for $|y|<1$.
We see that $W_\bullet(z)$ is of the form (\ref{eq:pointeddisk}) with
\begin{equation*}
c_{\pm} = z^{\diamond} \pm 2\sqrt{\zp}.
\end{equation*}
The stated bounds on $c_\pm$ follow directly from the fact that by construction $\zd=0$ if and only if $q_{k}=0$ for all odd $k$, while $\zp=Z_\bullet^++1>1$.
\end{proof}

\subsection{Proof of Proposition \ref{thm:walksequence}}\label{sec:proofwalksequence}

\begin{proof}
Notice that we can rewrite (\ref{eq:fdotdef}) and (\ref{eq:fdiamonddef}) in terms of the generating function $U'(z)$ in (\ref{eq:potential}) as 
\begin{align}
f^\bullet (x,y) &= [z^{-1}] \frac{z-y}{2x}\frac{U'(z)}{\sqrt{(z-y)^2-4x}},\label{eq:genfdot}\\
f^\diamond(x,y) &= [z^{-1}] \frac{U'(z)}{\sqrt{(z-y)^2-4x}}.\label{eq:genfdiamond}
\end{align}
If we set $c_\pm = \zd \pm 2\sqrt{\zp}$ and $r=-c_-/c_+$, we may rewrite $f^\diamond(\zp,\zd)$ in terms of $c_+$, $r$ and $\nu(k)$ for $k\geq -1$ using (\ref{eq:nudef}),
\begin{align*}
 f^\diamond(\zp,\zd)& = [z^{-1}]\frac{U'(z)}{\sqrt{(z-c_+)(z+c_+r)}}= [y^{-1}]\frac{U'(c_+y)}{\sqrt{(y-1)(y+r)}} \nonumber\\
 &= [y^{-1}]\sum_{k=1}^{\infty} q_k (yc_+)^{k-1}\sum_{l=0}^{\infty} h_r^{(0)}(l) y^{-l-1}= c_+\sum_{k=-1}^{\infty} \nu(k) h_r^{(0)}(k+1).
\end{align*}
Since $\zd=c_+(1-r)/2=c_+h_r^{(0)}(1)$ the second equation in (\ref{eq:zpzdeq}) is equivalent to
\begin{equation}\label{eq:harm1}
\sum_{k=-1}^{\infty} \nu(k) h_r^{(0)}(k+1) = h_r^{(0)}(1).
\end{equation}
Similarly we may evaluate
\begin{align*}
f^\bullet(\zp,\zd) &= [z^{-1}]\frac{z-\zd}{2\zp}\frac{U'(z)}{\sqrt{(z-c_+)(z+r c_+)}} \nonumber\\
&= [y^{-1}] \frac{8}{(1+r)^2}\left(y-\frac{1-r}{2}\right)\frac{U'(c_+y)/c_+}{\sqrt{(y-1)(y+r)}}\\
&=\frac{8}{(1+r)^2}\sum_{k=-1}^{\infty}\nu(k)\left(h_r^{(0)}(k+2)-\frac{1-r}{2}h_r^{(0)}(k+1)\right).\nonumber
\end{align*}
In combination with (\ref{eq:harm1}) the first equation in (\ref{eq:zpzdeq}) is equivalent to
\begin{align*}
\sum_{k=-1}^{\infty}\nu(k)h_r^{(0)}(k+2) &= \frac{(1+r)^2}{8}\left(1-\frac{16}{(1+r)^2c_+^2}\right) + \frac{1-r}{2}h_r^{(0)}(1) \nonumber\\
&= \frac{1}{8}(3-2r+3r^2)-\frac{2}{c_+^2} = h_r^{(0)}(2) - \nu(-2),
\end{align*}
where we used $\nu(-2)=2/c_+^2$.
Since $h_r^{(0)}(0)=1$, this is equivalent to
\begin{equation}\label{eq:harm2}
\sum_{k=-2}^{\infty}\nu(k)h_r^{(0)}(k+2) = h_r^{(0)}(2).
\end{equation}

Let us introduce the notation
\begin{equation*}
A_k := [z^{-1}]\frac{(z-\zd)^kU'(z)}{\left((z-\zd)^2-4\zp\right)^{3/2}},
\end{equation*}
which are nonnegative numbers for $k\geq 0$.
Assuming $\qseq$ is non-bipartite, one may evaluate $\mathfrak{M}_{\qseq}(\zp,\zd)$ using (\ref{eq:genfdot}) and (\ref{eq:genfdiamond}) to
\begin{equation*}
\mathfrak{M}_{\qseq}(\zp,\zd) = \begin{pmatrix}
0 & 0 & \zp-1 \\
2 \frac{\zp}{\zd}A_0 & A_1 & 0 \\
\frac{1-\zp(1-A_1)}{\zp-1} & 2\frac{\zp\zd}{\zp-1}A_0 & 0
\end{pmatrix},
\end{equation*}
where in the calculation of $\partial_xf^\bullet(\zp,\zd)$ we used that (\ref{eq:zpzdeq}) implies that $A_3-4\zp A_1=2\zp-2$.
 
Let $p(\lambda)=\det(\mathfrak{M}_{\qseq}-\lambda I)$ be the characteristic polynomial of the matrix $\mathfrak{M}_{\qseq}(\zp,\zd)$.
We claim that the statement that the spectral radius $\rho_\qseq$ is less or equal to one, which means that $p(\lambda)$ has no roots greater than one, is equivalent to $A_1+2\sqrt{z^+}A_0 \leq 1$.
In the case that $A_1\geq 1$ both statements are false, since $p(A_1)=(2\zp A_0)^2 > 0$, so let us assume $A_1< 1$.
In that case $A_1+2\sqrt{z^+}A_0 \leq 1$ is equivalent to $p(1)=\zp(4\zp A_0^2-(1-A_1)^2)\leq 0$, which is a necessary condition for the spectral radius to be less or equal to one.
It is also sufficient because one can check that $p(\lambda)-(2\zp A_0)^2=0$ has three distinct roots smaller than one, namely $A_1$ and $\pm\sqrt{1-\zp(1-A_1)}$.

The condition $A_1+2\sqrt{z^+}A_0 \leq 1$ is also the right one in the bipartite case, since then $A_0=0$ while $(\zp)^2\partial_xf^\bullet\leq1$ is equivalent to $A_1\leq 1$.
We conclude that in general $\qseq$ is admissible if and only if for some solution $(\zp,\zd)$ to (\ref{eq:zpzdeq}), or equivalently for some solution $(c_+,r)$ to (\ref{eq:harm1}) and (\ref{eq:harm2}), the following condition is satisfied:
\begin{align*}
1 \geq& A_1+2\sqrt{\zp}A_0 = (\partial_y +\sqrt{\zp}\partial_x)f^\diamond(\zp,\zd) = [z^{-1}]\frac{(z-\zd+2\sqrt{\zp})U'(z)}{((z-\zd)^2-4\zp)^{3/2}}\nonumber\\
&= [z^{-1}] \frac{U'(z)}{(z-c_+)^{3/2}\sqrt{z-c_-}} = [y^{-1}]\frac{U'(y c_+)/c_+}{(y-1)^{3/2}\sqrt{y+r}}=\sum_{k=0}^{\infty}\nu(k)h_r^{(1)}(k+1).
\end{align*}
\end{proof}

\subsection{Proof of Lemma \ref{thm:volumeexpansion}}\label{sec:proofvolumeexpansion}

\begin{proof}
In the notation of the other proofs in this appendix let us write $c_{\pm} = c_{\pm}(g=1) = \zd \pm 2\sqrt{\zp}$. 
Define the function $\mathbf{f} : \R^2 \to \left(\R\bigcup\{\infty\}\right)^2$ in terms of the functions $f^\diamond$ and $f^\bullet$ in (\ref{eq:fdotdef}) and (\ref{eq:fdiamonddef}) associated to $\qseq$ by
\begin{equation*}
\mathbf{f}(x,y) = \left(1-y^{-1}f^{\diamond}(x,y), 1-x+xf^\bullet(x,y)\right).
\end{equation*}
Applying Proposition \ref{thm:miermont} to the weight sequence $\qseq_g$ with $0<g\leq 1$ we find that there exists a solution to $\mathbf{f}(x,y)=(0,1-g)$.
From Proposition \ref{thm:walksequence} one may deduce that there is a unique such solution such that
\begin{equation}\label{eq:admcondexpansion}
(\partial_y + \sqrt{x} \partial_x)f^{\diamond}(x,y) \leq 1
\end{equation}
and the values $c_{\pm}(g)$ are given in terms of this solution by
\begin{equation}\label{eq:cpmexpansion}
c_{\pm}(g) = \frac{1}{\sqrt{g}}(y \pm 2\sqrt{x}).
\end{equation}

It follows from the regular criticality of $\qseq$ that the function $\mathbf{f}$ is analytic in the neighborhood of $(\zp,\zd)$, where it takes value $\mathbf{f}(\zp,\zd)=(0,0)$.
So in order to find the solutions $c_{\pm}(g)$ for $g$ close to $1$ it suffices to Taylor expand $\mathbf{f}$ around $(\zp,\zd)$.
The first-order partial derivatives, which appear in the matrix $\mathfrak{M}_{\qseq}(\zp,\zd)$ in (\ref{eq:mmatrix}), were already computed in the proof of Proposition \ref{thm:walksequence}.

Let us treat the bipartite and non-bipartite bases separately, starting with the latter.
Writing $\Delta x := x - \zp$ and $\Delta y := y - \zd$ we find after some straightforward manipulation
\begin{equation*}
\mathbf{f}(\zp+\Delta x,\zd+\Delta y) = - 2 A_0 (\Delta x - \sqrt{\zp}\Delta y)\, \left( \frac{1}{\zd}, \sqrt{\zp} \right) + \mathcal{O}(\Delta x^2+\Delta y^2), 
\end{equation*}
where we used that the criticality of $\qseq$ implies that $A_1+2\sqrt{\zp}A_0=1$.
It follows that $\lim_{g\uparrow 1} \Delta x/\Delta y = \sqrt{\zp}$ and that $\rho(x,y):=(-\zd\sqrt{\zp},1)\cdot\mathbf{f}(x,y)$ is stationary at $\Delta x = \Delta y = 0$.

One may check explicitly using the generating function identities (\ref{eq:genfdot}) and (\ref{eq:genfdiamond}) that
\begin{align*}
\frac{1}{2\sqrt{\zp}} (\partial_y + \sqrt{\zp} \partial_x)^2 \rho(x,y)\Bigg|_{\substack{x=\zp \\y=\zd}}\!\!\!\!\! &=\Big[ \compactfrac{1}{2} (\partial_y + \sqrt{\zp} \partial_x)^2f^\diamond(x,y) + \compactfrac{1}{2} \sqrt{\zp}(\partial_y + \sqrt{\zp} \partial_x)^2f^\bullet(x,y) \nonumber\\
&\quad+ (\partial_y + \sqrt{\zp} \partial_x)f^\bullet(x,y)\Big]_{\substack{x=\zp \\y=\zd}}\\
&=(\partial_y + \sqrt{x} \partial_x)^2f^\diamond(x,y)\Big|_{\substack{x=\zp \\y=\zd}}\nonumber
\end{align*}
With the help of (\ref{eq:potential}), (\ref{eq:hfundef}), and (\ref{eq:lenconst}), the latter may be evaluated to
\begin{align}
(\partial_y + \sqrt{x} \partial_x)^2f^\diamond(x,y)\Big|_{\substack{x=\zp \\y=\zd}}&=3[z^{-1}] \frac{U'(z)}{(z-c_+)^{5/2}\sqrt{z-c_-}}=\frac{3}{c_+} [y^{-1}]\frac{U'(yc_+)/c_+}{(y-1)^{5/2}\sqrt{y+r}}\nonumber\\
&= \frac{3}{c_+}\sum_{k=1}^{\infty}h_r^{(2)}(k+1)\nu(k) = \frac{3}{c_+} \lenconst_\nu.\label{eq:ddfdiamond}
\end{align}
Therefore, using $\sqrt{\zp}=(1+r)c_+/4$, we find that to second order in $\Delta y$ we have
\begin{equation*}
(-\zd\sqrt{\zp},1)\cdot\mathbf{f}(\zp+\sqrt{\zp}\Delta y,\zd+\Delta y) = \frac{3}{4}(1+r)\lenconst_\nu \Delta y^2 + \mathcal{O}(\Delta y^3).
\end{equation*}
It follows that for $g$ close enough to $1$ there are two solutions to $\mathbf{f}(x,y)=(0,1-g)$, which are given to leading order by
\begin{equation}\label{eq:xysol}
(x,y) = (\zp,\zd) \pm \sqrt{\frac{1-g}{\frac{3}{4}(1+r)\lenconst_\nu}}(\sqrt{\zp},1) + \mathcal{O}(1-g).
\end{equation}
It remains to check that one of them satisfies (\ref{eq:admcondexpansion}). 
Using (\ref{eq:ddfdiamond}) we find
\begin{align*}
(\partial_y + \sqrt{x}\partial_x)f^\diamond(x,y) &= 1 + \Delta y \left.(\partial_y + \sqrt{x}\partial_x)^2f^\diamond(x,y)\right|_{\substack{x=\zp\\ y=\zd}} + \mathcal{O}(1-g)\nonumber\\
&= 1 \pm \sqrt{\frac{12 \lenconst_\nu(1-g)}{(1+r)c_+^2}}  + \mathcal{O}(1-g), 
\end{align*}
which in case of a minus sign is smaller or equal to $1$ for $g$ close enough to $1$.
Finally, combining (\ref{eq:cpmexpansion}) with the version of (\ref{eq:xysol}) with a minus sign gives
\begin{equation}\label{eq:cpmexpansion2}
c_{+}(g) = c_{+} + 2\sqrt{\frac{1-g}{\frac{3}{4}(1+r)\lenconst_\nu}} + \mathcal{O}(1-g),\quad\text{while}\quad c_-(g) = c_- + \mathcal{O}(1-g),
\end{equation}
which finishes the proof in the non-bipartite case.

If $\qseq$ is bipartite, then $y=0$ for any $0<g\leq 1$. 
Therefore we only need to consider the second component $\mathbf{f}_2(x,0) = 1-x+xf^{\bullet}(x,0)=1-g$ of the equation above.
When $\qseq$ is critical $\mathbf{f}_2(x,0)$ is stationary at $x=\zp$ and one may check that
\begin{equation*}
\mathbf{f}_2(\zp+\Delta x,0) = \frac{3}{2\zp} \lenconst_\nu \Delta x^2 + \mathcal{O}(\Delta x^3).
\end{equation*}
This implies that for $g$ close to $1$ there are two solutions $x = \zp \pm \sqrt{\frac{2}{3}(1-g)\zp/\lenconst} + \mathcal{O}(1-g)$, of which the smaller one satisfies $(\partial_y+\sqrt{x}\partial_x)f^\diamond(x,0) = \partial_x\mathbf{f}_2(x,0) + 1 \leq 1$.
It follows that $c_{+}(g)$ obeys the same formula (\ref{eq:cpmexpansion2}) as in the non-bipartite case, while $c_-(g)=-c_+(g)$.
\end{proof}

\subsection{Proof of Lemma \ref{thm:rovdisk}}\label{sec:proofrovdisk}

\begin{proof}
We already noticed that the radius of convergence was at least equal to one, so we only need to show that it cannot be larger.
The regular critical case is covered by Lemma \ref{thm:volumeexpansion}, since it shows that $c_+(g)$ is non-analytic at $g=1$ and therefore the same holds for $W_g^{(l)}$. 

Let us assume that $\qseq$ is admissible and that there exists a $g_0>1$ such that $W_{g_0}^{(l)}(\qseq)<\infty$.
We claim that $\qseq$ can then not be non-regular critical, which would complete the proof.
Indeed, it follows from the identity $f^\diamond_{\qseq_g}( x, y ) = \frac{1}{\sqrt{g}}f^{\diamond}_{\qseq}( g x, \sqrt{g} y)$ that 
\begin{equation*}
f^\diamond_{\qseq}(g_0 \zp(\qseq_g), \sqrt{g_0} \zd(\qseq_g)) = \sqrt{g_0} \zd(\qseq_{g_0}) < \infty.
\end{equation*}
Since they are generating functions in $g$ of non-negative quantities, $\zp(\qseq_g)$ and $\zd(\qseq_g)$ are necessarily non-decreasing functions of $g \in [0,g_0]$.
Therefore there exists an $\epsilon > 0$ such that $f^\diamond_{\qseq}((1+\epsilon)\zp(\qseq),(1+\epsilon) \zd(\qseq)) < \infty$, but this is precisely Miermont's criterion for $\qseq$ being regular.
\end{proof}

\subsection*{Acknowledgments}
The author thanks Nicolas Curien and Jan Ambj\o rn for useful discussions and the anonymous referees for valuable suggestions.
The author acknowledges support from the ERC-Advance grant 291092,
``Exploring the Quantum Universe'' (EQU).

\bibliographystyle{hsiam} 

\end{document}